\documentclass[conference]{IEEEtran}

\IEEEoverridecommandlockouts

\usepackage[T1]{fontenc}
\usepackage{amsmath,amssymb,amsthm}
\usepackage{mathtools}
\usepackage{booktabs}
\usepackage{algorithm}
\usepackage[noend]{algorithmic}
\usepackage{xcolor}
\usepackage{tikz}
\newcommand{\oomit}[1]{}
\usetikzlibrary{external,automata,arrows.meta,positioning,calc,shapes,shadows,decorations,patterns,fit,backgrounds}

\tikzset{
	state/.style=
    {circle, draw, align=center, auto, initial text={}}, 
	>=stealth,
	loopright/.style={loop,looseness=5,out=35, in=-35},
	loopleft/.style={loop,looseness=5,out=145, in=215},
	loopabove/.style={loop,looseness=5,out=125, in=55},
	loopbelow/.style={loop,looseness=5,out=-125, in=-55}
} 

\usepackage[colorlinks=true]{hyperref}
\hypersetup{colorlinks=true,linkcolor=blue,citecolor=blue,urlcolor=blue}
\usepackage[capitalize,noabbrev]{cleveref}
\usepackage{capt-of}
\usepackage{microtype}
\usepackage{multirow}
\usepackage{array}
\usepackage{enumitem}
\usepackage{stmaryrd}

\allowdisplaybreaks

\theoremstyle{plain}
\newtheorem{theorem}{Theorem}
\newtheorem{lemma}[theorem]{Lemma}
\newtheorem{proposition}[theorem]{Proposition}
\newtheorem{corollary}[theorem]{Corollary}

\theoremstyle{definition}
\newtheorem{definition}{Definition}
\newtheorem{problem}{Problem}
\crefname{problem}{Problem}{Problems}
\newtheorem{example}{Example}

\theoremstyle{remark}
\newtheorem{remark}{Remark}

\AddToHook{env/lemma/begin}{\crefalias{theorem}{lemma}}
\AddToHook{env/proposition/begin}{\crefalias{theorem}{proposition}}
\AddToHook{env/corollary/begin}{\crefalias{theorem}{corollary}}

\usepackage{wrapfig}
\usepackage{graphicx} 
\usepackage{listings}
\usepackage{comment}

\usepackage{makecell}
\usepackage{diagbox}
\usepackage{tabularx}

\usepackage{pifont}

\usepackage[textsize=scriptsize, textwidth=6cm, obeyFinal]{todonotes}

\newcommand{\APorig}{\mathsf{AP}_{\mathsf{orig}}}
\newcommand{\APfin}{\mathsf{AP}_{\mathsf{final}}}

\newcommand{\calI}{\mathcal{I}}
\newcommand{\calIz}{\mathcal{I}_{0}}
\newcommand{\proj}{\pi}
\newcommand{\delay}{\Delta}
\newcommand{\Pairsconf}{\mathsf{Pairs}_{\mathsf{conf}}}
\newcommand{\Pairsundone}{\mathsf{Pairs}_{\mathsf{undone}}}
\newcommand{\Pairstgt}{\mathsf{Pairs}_{\mathsf{target}}}

\newcommand{\ltl}{\textmd{\textup{\textsf{LTL}}}}
\newcommand{\mitl}{\textmd{\textup{\textsf{MITL}}}}
\newcommand{\mitppl}{\textmd{\textup{\textsf{MITPPL}}}}

\newcommand{\tre}{\textmd{\textup{\textsf{TRE}}}}
\newcommand{\mtl}{\textmd{\textup{\textsf{MTL}}}}

\newcommand{\mightyppl}{\textsc{MightyPPL}}

\newcommand*{\AP}{\mathsf{AP}}
\renewcommand*{\P}{\mathsf{P}}

\newcommand{\F}{\mathcal{F}}

\newcommand{\until}{\mathbf{{U}}}

\newcommand{\since}{\mathbf{{S}}}

\newcommand{\R}{\mathbb{R}_{\geq 0}}
\newcommand{\N}{\mathbb{N}}

\renewcommand{\max}{\mathsf{max}}

\definecolor{saffron}{rgb}{1.0,0.49,0.0}

\makeatletter
\newcommand{\oset}[3][0ex]{%
  \mathrel{\mathop{#3}\limits^{
    \vbox to#1{\kern-2\ex@
    \hbox{$\scriptstyle#2$}\vss}}}}
\makeatother

\newcommand{\Pnkern}{%
  \mkern-2mu
}
\DeclareMathOperator{\eventually}{\mathbf{F}}
\DeclareMathOperator{\once}{\overset{\leftarrow}{\mathbf{F}}}
\DeclareMathOperator{\historically}{\overset{\leftarrow}{\mathbf{G}}}
\DeclareMathOperator{\globally}{\mathbf{G}}
\DeclareMathOperator{\nextx}{\mathbf{X}}
\DeclareMathOperator{\PnF}{\mathbf{P \Pnkern n}}
\DeclareMathOperator{\PnO}{\oset[-1pt]{\leftarrow}{\mathbf{P \Pnkern n}}}
\renewcommand{\F}{\eventually}

\makeatletter
\def\@IEEEsectpunct{.\ \,}
\def\paragraph{\@startsection{paragraph}{6}{\z@}{1.5ex plus 1.5ex minus 0.5ex}%
{0ex}{\normalfont\normalsize\bfseries}}
\makeatother

\begin{document}

\title{On Synthesis of Metric Interval Temporal Logics\thanks{The work of Khushraj Madnani was
supported by the ANRF ARG-MATRICS Grant (File No. ANRF/ARGM/2025/003152/MTR) under the project
titled ``Approaching the Decidability Purlieu of Timed Languages (ADePT).''}}
\author{
  \IEEEauthorblockN{Hsi-Ming Ho}
  \IEEEauthorblockA{
    University of Sussex, Brighton, UK\\
    hsi-ming.ho@sussex.ac.uk 
  }
  \and
  \IEEEauthorblockN{Shankara Narayanan Krishna}
  \IEEEauthorblockA{
    IIT Bombay, Mumbai, India\\
    krishnas@cse.iitb.ac.in
  }
  \and
  \IEEEauthorblockN{Khushraj Madnani}
  \IEEEauthorblockA{
    IIT Guwahati, Guwahati, India\\
    khushraj@iitg.ac.in 
  }
}

\maketitle

\begin{abstract}
Automated mining of formal specifications is vital for verifying real-time systems. However, existing passive learning approaches remain restricted to deterministic specifications or limited fragments of Timed Regular Expressions ($\tre{}$). To our knowledge, this paper presents the first framework to tackle \emph{precise} passive learning for an expressive timed logic---\emph{Metric Interval Temporal Logic} ($\mitl{}$)---without relying on predefined templates or restricted logic fragments.
 Our approach formally reduces the timed learning problem into a scalable untimed one. By identifying quantitative timing differences between positive and negative traces, we synthesise precise timed constraints and inject them as new Boolean atomic propositions. This embeds timing into the alphabet, delegating the complex formula evaluation to highly optimised, off-the-shelf untimed LTL tools.
 Crucially, our framework is complete, guaranteeing a separating specification can always be found. We evaluate our implementation across several benchmarks, demonstrating the effectiveness of our approach.
\end{abstract}


\section{Introduction}\label{sec:intro}

Formal specifications underpin the verification, runtime monitoring, and testing
of \emph{real-time} and \emph{cyber-physical systems} (CPS). However, manually crafting these
rigorous mathematical formulae is notoriously complex and prone to error. While
recent advances in \emph{Large Language Models} (LLMs) have begun to alleviate this
bottleneck by translating natural language requirements into formal specifications (see, e.g.,~\cite{xie2025effective}), this
approach still relies on the existence of clear, comprehensive textual
comments or documentation. In many practical scenarios, engineers instead possess vast
amounts of implicit behavioural data, such as execution traces, logs, and
counterexamples. \emph{Specification mining} (or `\emph{learning}') addresses the inverse problem: synthesising a 
mathematical model that directly captures the behaviours hidden within 
samples of labelled traces, namely the positive 
and negative trace sets, $\Omega^{+}$ and $\Omega^{-}$, respectively. 

For (untimed) \emph{Linear Temporal Logic} ($\ltl{}$), the mining problem is well
studied. Methods such as~\cite{neider2018learning, arif2020syslite, raha2022scalable} are implemented in tools that handle thousands of traces. However, in many
practical applications---such as autonomous driving,
medical monitoring, and industrial automation---system correctness depends not
only on the logical order of events but also on the strict time bounds of their
occurrences. To rigorously reason about these timing constraints, one must turn
to timed logics such as \emph{Metric Interval Temporal Logic} ($\mitl{}$)~\cite{alur1996benefits}, a quantitative extension of $\ltl{}$ with integer-bounded timing constraints.
However, the learning problem for $\mitl{}$ is conceivably harder: constructing a
consistent $\mitl{}$ formula requires choosing both a structure (nested operators, atomic propositions) \emph{and} interval parameters.

A na\"ive attempt to reuse mature $\ltl{}$ learning tools is the simple
`untiming' reduction: for a \emph{timed word}
$\omega = (d_1,\sigma_1)\cdots(d_\ell,\sigma_\ell)$ over $\Sigma_\AP = 2^\AP$ (i.e.~$\sigma_i$ is a
subset of \emph{atomic propositions}
$\AP$ for all $i$, $1 \leq i \leq \ell$,  
and $d_i$'s denote the time delay before witnessing $\sigma_i$'s),
let $\mu(\omega) = \sigma_1\cdots\sigma_\ell$
be the sequence obtained from $\omega$ by removing the `delay' component. 
We then feed
$\{\mu(\omega) : \omega \in \Omega^+\}$ and
$\{\mu(\omega) : \omega \in \Omega^-\}$ to an off-the-shelf $\ltl{}$ learner.
This approach would \emph{not} work, however, whenever a positive and a negative trace 
share an identical untimed projection, as obviously no $\ltl{}$ formula can distinguish between them.
In this case, the two traces are said to be \emph{propositionally indistinguishable}
with respect to $\AP$.

\begin{example}\label{ex:naive-fails}
Consider an EDF (Earliest Deadline First) scheduler with two tasks $T_1,T_2$ where $T_i$ 
involves releasing a job (witnessed as $p_i$) and its completion (witnessed as $q_i$).
The trace
$\omega^+ = (0, \{p_1\})(3.5, \{q_1\})(4, \{p_2\})(8.2, \{q_2\})$ records a successful release
($p_i$) and completion ($q_i$) sequence, while the trace
$\omega^- = (0, \{p_1\})(12.5, \{q_1\})(4, \{p_2\})(8.2,\{q_2\})$ records the same sequence of
events but with $T_1$ completing after a delay of $12.5$ time units, well past its deadline.
The untimed projection of $\omega^{+}, \omega^{-}$ are identical: $\mu(\omega^+) = \mu(\omega^-) = \{p_1\}\,\{q_1\}\, \{p_2\}\, \{q_2\}$, and obviously
no $\ltl{}$ formula over $\AP = \{p_1, q_1, p_2, q_2\}$ distinguishes them.
\end{example}

\paragraph{Quantitative distinguishability}
When two propositionally indistinguishable traces are nonetheless
behaviourally distinct (i.e.~one in $\Omega^{+}$ and one in $\Omega^{-}$), their divergence is purely quantitative: some
pair of event occurrences in the two traces differs in delay. We
formalise this as a \emph{pivot}---a tuple $((i, j), c)$ such that 
one trace's inter-event delay between positions $i$ and $j$ is below the integer boundary $c$ and the other's is above $c$.
Pivots are easily computable certificates of \emph{quantitative distinguishability}.
The crux of our framework is that a formula such as $\eventually_{[0,c]}p$
or $\once_{[0,c]}p$ evaluated at the corresponding position can distinguish
the two traces involved. When treated as a fresh proposition
and added to the set of atomic propositions, this formula serves as a new \emph{feature} and effectively restores
propositional distinguishability. 

\paragraph{From timed mining to untimed mining}
The feature-engineering pass is
guaranteed to make progress on every iteration.
Once enough synthesised features have been injected into the set of atomic propositions, every
positive-negative trace pair $(\omega^+, \omega^-)$ are propositionally distinguishable with respect
to the updated (expanded) set of atomic propositions $\AP$. At that point, the set of positive and negative traces can both be recasted onto the updated $\AP$ with delays discarded, and an untimed $\ltl{}$ learner
applied with no further modification. The expressive power lost when
discarding delays is recovered by the features the learner now sees in
the set of atomic propositions. In other words, this reduction turns a hard \emph{timed} mining problem into an
easier \emph{untimed} one.

\paragraph{Contributions}
We make the following contributions.
\begin{enumerate}[leftmargin=*]
  \item \textbf{Formalisation of propositional indistinguishability
    (\Cref{sec:pind}).} We define propositional
    indistinguishability with respect to arbitrary subsets of atomic propositions, prove a structural
    characterisation of when a pair is quantitatively distinguishable but
    propositionally indistinguishable (\Cref{thm:pivot-char}), and reduce
    such cases to a finite collection of integer pivots.
  \item \textbf{Subset-projection collision detection
    (\Cref{sec:phase1}).} We give an algorithm that enumerates
    subsets of atomic propositions in order of decreasing size and identifies every
    propositionally indistinguishable trace pair together with all the pivots
    which explain \emph{how the two traces may be distinguished}. The algorithm runs in
    $\mathcal{O}(2^{|\AP|}\cdot N^2 \cdot \ell^2)$ in the worst case (where $N$ is the number of traces in the larger of $\Omega^+, \Omega^-$ and $\ell$ is the length of the longest trace) but is highly
    effective in practice because most real-world specifications use small
    numbers of atomic propositions.
  \item \textbf{Layered feature-synthesis framework
    (\Cref{sec:phase2,sec:phase3}).}
    We detail the successive phases of our translation pipeline: template injection, 
layered feature synthesis (formulated as a set-cover problem), and proposition 
minimisation. During feature synthesis, the engine primarily relies on 
standard $\mitl{}$ modalities ($\eventually_I$, $\once_I$) to extract separating 
timing constraints. Where these prove insufficiently expressive, such as for 
complex counting patterns, we use Pnueli modalities 
($\PnF_J$, $\PnO_J$) as a fallback strategy. Furthermore, domain-specific 
template injection can be applied prior to synthesis, acting as a structural 
pre-filter to accelerate the resolution of known specification patterns.
  \item \textbf{Soundness and completeness
    (Sec.~\ref{sec:theory}).} We prove that the synthesised features
    suffice to make every $(\omega^+, \omega^-) \in \Omega^{+} \times \Omega^{-}$ propositionally distinguishable
    (\Cref{thm:soundness}), and that the
    framework is complete in the sense that
    we can generate an $\mitppl{}$ formula that distinguishes$(\Omega^{+}, \Omega^{-})$ if they can be distinguished by a timed automaton (\Cref{thm:completeness}).

  \item \textbf{Implementation and evaluation (Sec.~\ref{sec:experiments}).}
  We implement the framework in Python and evaluate it across three specific
  benchmark scenarios: (1) learning the behaviours of a simple timed automaton,
  (2) learning typical parameterised specification patterns used in the online
  monitoring of cyber-physical systems, and (3) learning to determine the
  precise timing constants within a train-gate controller. We demonstrate the
  feasibility and computational efficiency of our translation pipeline,
  contextualising our results against the performance of direct $\mtl{}$ mining
  and $\tre{}$ synthesis reported in~\cite{raha2023synthesizing, wang2025synthesis},
  and explicitly evaluate the structural trade-off between solver scalability and out-of-sample generalisation. 
  
\end{enumerate}

\paragraph{Outline}
\Cref{sec:example} presents a motivating example. \Cref{sec:prelim}
collects preliminaries on timed words, $\mitl{}$ and $\mitppl{}$, untimed $\ltl{}$
learning, etc.
\Cref{sec:pind} formalises
propositional indistinguishability and present the four phases of the process.
\Cref{sec:theory} establishes the framework's correctness, progress, and
complexity.
\Cref{sec:experiments} reports the experimental evaluation.
\Cref{sec:related,sec:conclusion} discuss related work and conclude.

\section{Motivating Example}\label{sec:example}

We motivate the framework using an example adapted from~\cite{wang2025synthesis}.
In their work, certain problem instances are
explicitly identified as out-of-scope (termed `obscured') because
separating them inherently requires the \emph{intersection} operator in
\emph{timed regular expressions} ($\tre{}$s).  We demonstrate how our framework
handles this trivially.

\paragraph{Setup} Consider a simple system with a single event type,
giving the original set of atomic propositions $\APorig=\{p\}$. The goal is to
distinguish $\Omega^+ = \{ \omega^+_1 \}$  from $\Omega^- = \{ \omega^-_1, \omega^-_2 \}$
with a simple $\tre{}$ (without intersection).
\cite{wang2025synthesis} demonstrates that while the positive trace can be separated from each negative trace
individually using simple $\tre{}$s, distinguishing it from both
simultaneously requires intersection (cf.~\cite{TRE}), which their synthesis method
explicitly forbids and cannot support.
Specifically, the traces are given as follows:
\[
\begin{aligned}
\omega^+_1 &= (0, \emptyset)(1.5, \{p\})(2.6, \{p\})(1.5, \{p\})\\
\omega^-_1 &= (0, \emptyset)(1.2, \{p\})(2.6, \{p\})(1.5, \{p\})\\
\omega^-_2 &= (0, \emptyset)(1.5, \{p\})(2.6, \{p\})(1.2, \{p\}) 
\end{aligned}
\]
All three traces share the exact same untimed projection $\emptyset\,\{p\}\,\{p\}\, \{p\}$,
making them propositionally indistinguishable 

\paragraph{What our framework does}
Phase~1 (Sec.~\ref{sec:phase1}) detects the propositionally indistinguishable
pairs and extracts pivots:
\begin{itemize}[leftmargin=*]
  \item Pair $(\omega^+_1, \omega^-_1)$: the sum of the first two delays is
  $4.1$ time units in the positive trace and $3.8$ in the negative trace, crossing the
  integer boundary $4$.
  \item Pair $(\omega^+_1, \omega^-_2)$: the sum of the last two delays is
  $4.1$ time units in the positive trace and $3.8$ in the negative trace, again crossing
  the integer boundary $4$.
\end{itemize}
Rather than attempting to construct a complex monolithic formula with intersections,
Phase~3 (\Cref{sec:phase3}) synthesises two new 
$\mitl{}$ formulae that serve as features: 
\[
\phi_1 = \once_{[0,4]} (\neg p), \; \phi_2 = \once_{[0,4]} (\once_{[0,4]} (\neg p)) 
\]
The meaning of $\phi_1$ is: `\emph{there is an event with $\neg p$ in at most $4$ time units in the past from now}'. This holds at position $3$ in $\omega^-_1$, but not at position $3$ in $\omega^+_1$ and $\omega^-_2$. 
Similarly, $\phi_2$ holds at position $4$ in $\omega^-_1$ and $\omega^-_2$, 
but not at position $4$ in $\omega^+_1$.

\paragraph{Expanded atomic propositions and untimed reduction}
The expanded set of atomic propositions is $\AP = \APorig \cup \{\phi_1,\, \phi_2\}$.
The positive and negative traces, with all delays stripped, are now pairwise distinguishable as untimed strings:
\[
\small 
\begin{aligned}
\mu^+_1 &= (0, \{\phi_1, \phi_2\})(1.5, \{p, \phi_1, \phi_2\})(2.6, \{p, \phi_2\})(1.5, \{p\})\\
\mu^-_1 &= (0, \{\phi_1, \phi_2\})(1.2, \{p, \phi_1, \phi_2 \})({2.6, \{{{p},{\phi_1}, {\phi_2}}\}})(1.5, \{p, \phi_2\})\\
\mu^-_2 &= (0, \{\phi_1, \phi_2\})(1.5, \{p, \phi_1, \phi_2\})(2.6, \{p, \phi_2\})(1.2, \{p, \phi_2\}) 
\end{aligned}
\]
The Boolean
evaluations of $\phi_1$ and $\phi_2$ at each position natively separate the
behaviours around the critical $4$ time unit boundary. 
When fed to the off-the-shelf $\ltl{}$ learner \textsc{Bolt}~\cite{bathie2026ltl}, the tool
outputs the formula

\[
\eventually (\neg \phi_2) \equiv \eventually (\neg \once_{[0,4]} (\once_{[0,4]} (\neg p))) \;,
\]
seamlessly resolving the exact problem instance that precludes the existing
$\tre{}$-based synthesis method of~\cite{wang2025synthesis}.

\section{Preliminaries}\label{sec:prelim}

\paragraph{Timed words and projections}
Let $\AP$ be a finite set of atomic propositions, and let $\Sigma_{\AP} = 2^{\AP}$ denote the corresponding finite alphabet. Let $\N$ and $\R$ denote the sets of natural numbers and non-negative real numbers, respectively.

\begin{definition}[Timed Word]\label{def:timed-word}
A \emph{timed word} (or a \emph{trace}) over $\Sigma_{\AP}$ is a finite sequence $\omega =
(d_1,\sigma_1)\cdots(d_\ell,\sigma_\ell) \in (\R \times \Sigma_{\AP})^*$, where each
$d_i$ represents the non-negative time delay elapsed immediately prior to
observing the symbol $\sigma_i$. We denote the length of $\omega$ as
$|\omega| = \ell$. 
The \emph{untimed projection} of $\omega$ is the sequence of symbols
$\mu(\omega)=\sigma_1\cdots\sigma_\ell$. For a set $\Omega$ of timed words,
we write $\mu(\Omega)$ for $\bigcup_{\omega \in \Omega} \{ \mu(\omega) \} $.
\end{definition}

\begin{definition}[Propositional Projection]\label{def:projection} For a subset
of atomic propositions $\P \subseteq \AP$, the \emph{$\P$-projection} of a
letter $\sigma \in \Sigma_{\AP}$ is $\proj_\P(\sigma) = \sigma \cap \P$. The
$\P$-projection of an untimed string $\mu = \sigma_1\cdots\sigma_\ell$ is the
string formed by projecting each individual letter in place, yielding
$\proj_\P(\sigma_1)\cdots\proj_\P(\sigma_\ell)$. Similarly, the $\P$-projection of
a timed word $\omega = (d_1, \sigma_1)\cdots(d_\ell, \sigma_\ell)$, denoted
$\proj_\P(\omega)$, is the timed word $(d_1, \sigma_1 \cap \P)\cdots(d_\ell, \sigma_\ell
\cap \P)$.
For a set $\Omega$ of timed words,
we write $\proj_\P(\Omega)$ for $\bigcup_{\omega \in \Omega} \{ \proj_P(\omega) \} $.
\end{definition}

\begin{example}[Propositional Projection]\label{ex:projection} Let
\[
 \omega = (0.5, \{p_1, r_1\})(3.2, \{p_2,
q_1\})(3.0, \{q_2\})(4.5, \{p_1, q_2\})
\]
be a timed word over $\Sigma_\AP$ where $\AP = \{p_1, p_2, q_1, q_2, r_1\}$.
The $\{p_1, p_2\}$-projection of $\omega$ is:
$$ \proj_{\{p_1,p_2\}}(\omega) = (0.5,
\{p_1\})(3.2, \{p_2\})(3.0, \emptyset)(4.5, \{p_1\}) $$
The $\{q_1, q_2\}$-projection of $\omega$ is:
$$ \proj_{\{q_1,q_2\}}(\omega) = (0.5, \emptyset)(3.2,
\{q_1\})(3.0, \{q_2\})(4.5, \{q_2\}) $$
\end{example}

\paragraph{$\mitl{}$ with Past and Pnueli modalities}
We use the standard $\mitppl{}$ syntax and semantics
of~\cite{ho2026MightyPPL}. Let $\langle$ denote a left-open
``('' or left-closed ``['' boundary, and let $\rangle$ denote a right-open ``)''
or right-closed ``]'' boundary. Let $\calI$ be the set of all intervals
$\langle l, u \rangle$ where $l \leq u$, $l \in \N$, and $u \in \N \cup
\{\infty\}$. Likewise, let $\calIz = \{[0,u\rangle \mid u \in \N \cup
\{\infty\} \}$. An
interval is \emph{singular} if it takes the form $[c,c]$ for $c \in \N$, and
\emph{non-singular} otherwise.

\begin{definition}[$\mitppl{}$ Syntax]\label{def:mitppl}
$\mitppl{}$ formulae over $\AP$ are defined by the grammar:
\begin{align*}
\varphi ::={}& \mathbf{true} \mid p \mid \neg\varphi \mid \varphi_1 \wedge \varphi_2
              \mid \varphi_1 \until_I \varphi_2 \mid \varphi_1 \since_I \varphi_2\\
            & \mid \PnF_J(\varphi_1,\dots,\varphi_k) \mid \PnO_J(\varphi_1,\dots,\varphi_k),
\end{align*}
where $p \in \AP$, $I \in \calI$ is a non-singular interval, $J \in \calIz$ is a bounded interval, and $k \in \N$ with $k \geq 2$. 

Standard derived modalities include: $\F_I\varphi = \mathbf{true}\,\until_I\,\varphi$, $\globally_I\varphi = \neg\F_I\neg\varphi$, $\once_I\varphi = \mathbf{true} \,\since_I\,\varphi$, and $\historically_I\varphi = \neg\once_I\neg\varphi$.
\end{definition}

Given a timed word $\omega=(d_1,\sigma_1)(d_2,\sigma_2)\dots$ over $\Sigma_{\AP}$, a position $i \in \N$, and any $a \in \AP$, we define the \emph{pointwise-semantics} of $\mitppl{}$ inductively:
\[ \begin{array}{lclcl}
  \omega, i & \models & p & \text{iff} & p \in \sigma_i. \\
  \omega, i & \models & \varphi_1 \wedge \varphi_2 & \text{iff} & \omega, i \models \varphi_1 \text{ and } \omega, i \models \varphi_2. \\
  \omega, i & \models & \varphi_1 \vee \varphi_2 & \text{iff} & \omega, i \models \varphi_1 \text{ or } \omega, i \models \varphi_2. \\
  \omega, i &  \models & \nextx_I \varphi         &\text{iff} & \omega, i+1 \models \varphi \text{ and } d_{i+1} \in I.\\
  \omega, i &  \models & \neg \varphi          &\text{iff} & \omega, i \nvDash \varphi.
\end{array}
\]
\[
{\small
\begin{array}{w{l}{2ex}w{c}{1ex}w{l}{4em}w{c}{1ex}l}
  \omega, i & \models & \varphi_1 \until_I \varphi_2 & \text{iff} & {\footnotesize \left\{ \begin{array}{l}\exists j \geq  i \text{ s.t. } 
              \sum_{k=i+1}^j d_k\in I \text{, } \\
              \omega, j \models \varphi_2 \text{, } \text{and } \forall i \leq  k < j, \omega, k  \models
    \varphi_1.\end{array} \right.} \\
    ~&~&~&~&~\\
  \omega, i & \models & \varphi_1 \since_I \varphi_2 & \text{iff} & {\footnotesize \left\{ \begin{array}{l}\exists  j \leq  i \text{ s.t. }
              \sum_{k=j+1}^i d_k \in I \text{, } \\
              \omega, j \models \varphi_2 \text{, } \text{and } \forall j < k  \leq i, \omega, k
     \models \varphi_1.\end{array} \right.}
\end{array}
}
\]
\[
{\footnotesize
    \begin{array}{w{l}{2ex}w{c}{1ex}w{l}{7em}w{c}{1ex}l}
          \omega, i &\models & \PnF_J(\varphi_1, \ldots, \varphi_k) &\text{iff}&
          \left\{\begin{array}{l}\exists j_k > j_{k-1} > \ldots > j_1 > i \\
\text{s.t. } \forall 1 \le m \le k, 
{\sum_{n=i+1}^{j_m} d_n \in J} \\ \text{and } \omega, i_m{\models}\varphi_m. \end{array}\right. \\
 & & & & \\
          \omega, i &\models & \PnO_J(\varphi_1, \ldots, \varphi_k) &\text{iff}&
          \left\{\begin{array}{l}\exists j_k < j_{k-1} < \ldots < j_1 < i \\
\text{s.t. } \forall 1\le m \le k, 
{\sum_{n=j_m + 1}^{i} d_n \in J} \\ \text{and } \omega, i_m{\models}\varphi_m. \end{array}\right.
\end{array}
}
\]
We write $\omega \models \varphi$ to denote $\omega,1 \models \varphi$. We omit the subscript, for succinctness, when the intended interval is $[0,\infty)$ as the interval imposes no restriction. Hence, $\ltl{}$ can be seen as a subclass of $\mitl{}$ with no intervals.

\paragraph{Passive untimed $\ltl{}$ learning} We adopt the standard passive
learning setup of~\cite{lemieux2015general,bathie2026ltl}. An $\ltl{}$
formula uses standard temporal operators ($\eventually$, $\globally$, $\mathbf{X}$,
$\until$, etc.), \emph{without}  timing intervals.

\begin{problem}[Untimed $\ltl{}$ Learning]\label{prob:untimedlearning} Given an
untimed sample $(\Omega^+, \Omega^-)$, where $\Omega^+$ and $\Omega^-$ are
sets of finite words over $\Sigma_{\AP}$, find an
$\ltl{}$ formula $\varphi$ such that $\omega \models \varphi$ for all $\omega
\in \Omega^+$ and $\omega \not\models \varphi$ for all $\omega \in \Omega^-$.
\end{problem}

Trivially, as long as $\Omega^+$ and $\Omega^-$  are strictly
disjoint, a distinguishing formula always exists (e.g., by explicitly encoding the
finite set $\Omega^+$ as a large disjunction).
However, such literal encodings severely
overfit and fail to capture meaningful underlying patterns.

In practice, modern $\ltl{}$ learners employ various algorithmic strategies to
systematically extract a minimum-size consistent formula. While some approaches
reduce the problem to a SAT or SMT encoding (e.g.,~\cite{neider2018learning}),
others leverage various syntactic and semantic enumeration techniques
(e.g.,~\cite{raha2022scalable,valizadeh2024gpu,bathie2026ltl}). Regardless of the specific
methodology, our framework treats the untimed $\ltl{}$ learner purely as a \emph{black
box}. We remain entirely agnostic to its internal mechanisms, relying solely on
the assumption that it returns a sound specification.

\paragraph{Passive $\mitppl{}$ learning}
We now lift the passive learning
problem to the timed setting. While the untimed setting abstracts away absolute
durations, the primary objective of our framework is to synthesise timing
constraints directly from the delay data in traces. We formalise this core task
as follows:

\begin{problem}[Timed $\mitppl{}$ Learning]\label{prob:timedlearning} Given a
timed sample $(\Omega^+, \Omega^-)$, where $\Omega^+$ and $\Omega^-$ are
sets of timed words over $\Sigma_{\AP}$,
find an $\mitppl{}$ formula $\varphi$ such that $\omega \models \varphi$ for all
$\omega \in \Omega^+$ and $\omega \not\models \varphi$ for all $\omega \in
\Omega^-$. \end{problem}

Our framework imposes no structural assumptions on the input traces; our theoretical
guarantees hold strictly regardless of trace length, duration, or event
density. Unlike the untimed setting, simply having disjoint positive and negative
trace sets is insufficient to guarantee the existence of a distinguishing formula.
As detailed in the subsequent sections, Problem~\ref{prob:timedlearning} admits
a valid solution if and only if the sets $\Omega^+$ and $\Omega^-$ can be
separated by a timed automaton.

\begin{remark}[Rational Constants and Scaling]
While our pivot extraction algorithm relies
on integer boundaries, non-integer timing constraints can always be 
accommodated via scaling. Provided the required granularity of
timing precision is known, which is virtually always the case in practice, one
can safely map the problem to the integer domain without any loss of precision.
For example, if a target system requires thresholds at $0.5$ granularity,
pre-multiplying all trace timestamps by $2$ safely satisfies the integer
assumption prior to learning.
\end{remark}

\section{The Synthesis Framework}\label{sec:pind}

In this section, we outline our four-phase translation architecture. Given a
sample of positive and negative timed words, our objective is to identify a set
of localised $\mitppl{}$ features such that, when injected as new atomic
propositions, \emph{the untimed projections of the positive and negative traces
become strictly disjoint}. We achieve this through four distinct phases: (1)
\emph{subset-projection collision detection}, which isolates the exact locations
where traces share identical logical structures but diverge in timing;
(2) \emph{template injection}, which applies user-provided domain knowledge to
rapidly filter out broad structural violations; (3) \emph{layered feature
synthesis}, which iteratively constructs and evaluates new temporal features to
separate the remaining conflicts using a global 
$\mathcal{O}(\ell)$ sliding-window algorithm (where $\ell$ is the length of the longest
trace); and (4) \emph{proposition minimisation}, which eliminates redundant
features to ensure the final set of propositions is as compact as possible.

 Before detailing these phases, we first formalise the obstacle
that defeats a naive reduction (i.e., the direct application of untimed
$\ltl{}$ learning) and precisely characterise the conditions under which it
occurs.

\begin{definition}[Propositional Indistinguishability]\label{def:p-indist}
Let $\omega_1, \omega_2$ be timed words of equal lengths (i.e.~$|\omega_1| = |\omega_2|$) over $\Sigma_\AP$ and $\mathsf{P} \subseteq \AP$.
The pair $(\omega_1, \omega_2)$ is \emph{propositionally indistinguishable with respect to $\P$}, 
written $\omega_1 \equiv_{\P} \omega_2$,
if
$\mu(\proj_\P(\omega_1)) = \mu(\proj_\P(\omega_2))$.
Note that this implies $\omega_1 \equiv_{\P'} \omega_2$ for all $\P' \subseteq \P$.
\end{definition}

\begin{lemma}[$\ltl{}$ Indistinguishability]\label{lem:ltl-limits}
If $\omega_1 \equiv_{\P} \omega_2$, 
then for every untimed $\ltl{}$
formula $\varphi$ over $\P$,
$\omega_1 \models \varphi \iff \omega_2 \models \varphi$.
\end{lemma}

\begin{corollary}\label{cor:ltl-cannot-separate}
If $\omega^+ \in \Omega^+$ and $\omega^- \in \Omega^-$ satisfy
$\omega^+ \equiv_{P} \omega^-$, then no untimed $\ltl{}$ formula can distinguish $\Omega^+$ from
$\Omega^-$.
\end{corollary}

\paragraph{Quantitative distinguishability}

The key structural observation relies on the notion of \emph{simple elementary languages}~\cite{waga2023active}, which establishes a necessary condition for distinguishability with respect to \emph{timed automata}~\cite{alur1994theory}.
By definition, a simple elementary language groups together propositionally indistinguishable timed words that satisfy the exact same tightest integer-bounded time constraints for all possible consecutive sequences of delays (i.e., the sum of the delays strictly falls between two consecutive integers $t$ and $t+1$, or is exactly equal to $t$). 

For example, consider the sequence of two delays $d_1, d_2$ associated
with an untimed event sequence $\{p\}\,\{q\}$. The timed words with delays $(0.4, 0.7)$
and $(0.3, 0.9)$ belong to the exact same simple elementary language because
their delays and sums share identical integer bounds: $d_1 \in (0, 1)$,
$d_2 \in (0, 1)$, and $d_1 + d_2 \in (1, 2)$. In contrast, a timed word
with delays $(0.6, 0.2)$ belongs to a different simple elementary language; although its
individual delays still fall within $(0, 1)$, their sum ($0.8$) strictly falls
into the distinct interval $(0, 1)$.

Consequently, if two traces are propositionally indistinguishable yet can be separated by a timed automaton (and by extension, any formal timed specification), they cannot belong to the same simple elementary language. This implies that their tightest integer bounds must differ, guaranteeing the existence of a specific integer time boundary---which we term a \emph{pivot}---at which the two traces' inter-event delays disagree.

\begin{definition}[Pivot]\label{def:pivot}
Let $\omega_1, \omega_2$ be propositionally indistinguishable with respect to
$\P \subseteq \AP$ (i.e.~$\mu(\proj_\P(\omega_1)) = \mu(\proj_\P(\omega_2))$)
and $|\omega_1| = |\omega_2| = \ell$.
A \emph{pivot} for $(\omega_1, \omega_2)$ is a triple $((i,j), c)$ where $1 \leq i < j \leq \ell$, $c \in \N$, and
\[
\bigl[\delay^1_{i,j} \leq c < \delay^2_{i,j}\bigr]
\;\;\vee\;\;
\bigl[\delay^2_{i,j} \leq c < \delay^1_{i,j}\bigr] \textit{ or}
\]
\[
\bigl[\delay^1_{i,j} < c \leq \delay^2_{i,j}\bigr]
\;\;\vee\;\;
\bigl[\delay^2_{i,j} < c \leq \delay^1_{i,j}\bigr]
\]
with $\delay^k_{i,j} = \Sigma_{m \in (i, j]} d^k_{m}$ the accumulated delay between the $i$th and
$j$th events of $\omega_k$ ($k \in \{1,2\}$).
\end{definition}

\begin{remark}[Pivot Multiplicity]\label{rem:pivot-multiplicity}
A pair of inter-event delays $\delay^1_{i,j} < \delay^2_{i,j}$
generates exactly $\lfloor \delay^2_{i,j} \rfloor - \lceil \delay^1_{i,j}
\rceil + 1$ integer pivots, but for our purpose we just need one pivot 
for $(i, j)$.
\end{remark}

\begin{definition}[Quantitative Indistinguishability]\label{def:q-indist}
Let $\omega_1, \omega_2$ be timed words of equal lengths (i.e.~$|\omega_1| = |\omega_2|$) over $\Sigma_\AP$ and $\P \subseteq \AP$.
The pair $(\omega_1, \omega_2)$ is \emph{quantitatively indistinguishable}, 
written $\omega_1 \equiv_{\textit{time}} \omega_2$,
if
$\proj_\P(\omega_1), \proj_\P(\omega_2) \in L$ for some simple elementary language $L$
over $\Sigma_P$.
\end{definition}

\begin{lemma}[Pivot Characterisation]\label{thm:pivot-char}
Let $\omega_1, \omega_2$ be propositionally indistinguishable with respect to
$\P \subseteq \AP$ but quantitatively distinguishable. Then there exists a pivot
$((i,j), c)$ for $(\omega_1, \omega_2)$.
\end{lemma}

\begin{example}[Pivots]\label{ex:pivots}
The pair $(\omega^+, \omega^-)$ of~\Cref{ex:naive-fails}
is propositionally indistinguishable with respect to any subset $\P$ of $\AP = \{p_1, q_1, p_2, q_2\}$
but quantitatively distinguishable.
The pivots are
$((1, 2), c)$ for $c \in \{4,5,\dots,12\}$,
$((1, 3), c)$ for $c \in \{8,9,\dots,16\}$,
and $((1, 4), c)$ for $c \in \{16,17,\dots,24\}$.
\end{example}

\subsection{Phase 1: Subset-Projection Collision Detection}\label{sec:phase1}

\Cref{alg:phase1} takes as input a sample $(\Omega^+, \Omega^-)$ over $\Sigma_{\APorig}$ and a user-specified threshold $k$. It yields two outputs: a set $\Pairsconf \subseteq \Omega^+ \times \Omega^-$ of \emph{conflicting trace pairs} that are propositionally indistinguishable with respect to at least one $\P \subseteq \AP$ with $|\P| \geq k$, and a map $\textsf{Mismatches}$ from each conflicting pair to its list of pivots.

The primary motivation for this subset-projection mechanism stems from the
observation that, in many practical scenarios, $\Omega^+$ and $\Omega^-$ are already propositionally distinguishable over the complete
set of atomic propositions $\AP$.
By systematically projecting the traces onto
restricted subsets of atomic propositions, the framework intentionally abstracts away these
irrelevant structural variations, effectively \emph{forcing} otherwise distinct
traces to become propositionally indistinguishable. 
Once this structural
equivalence is artificially induced, any remaining behavioural divergence
between the positive and negative traces must be inherently quantitative. This
purposeful isolation enables the algorithm to precisely pinpoint the pure timing
differences and extract their corresponding pivots.

\begin{algorithm}[t]
\caption{Phase 1: Subset-Projection Collision Detection}
\label{alg:phase1}
\begin{algorithmic}[1]
\REQUIRE $(\Omega^+, \Omega^-)$; $\APorig$; $k$, $0 \leq k \leq |\APorig|$
\ENSURE $\Pairsconf \subseteq \Omega^+ \times \Omega^-$; \\
      \hspace*{0.65cm}  $\textsf{Mismatches} : \Pairsconf \to 2^{\mathrm{Pivots}}$
\STATE $\Pairsconf \gets \emptyset$;\quad $\textsf{Mismatches}[\cdot] \gets \emptyset$
\FOR{$r = |\APorig|, |\APorig| - 1, \ldots, k $}
  \FORALL{$P \subseteq \APorig$ with $|P| = r$}
    \STATE Compute $\proj_P(\omega)$ for each $\omega \in \Omega^+ \cup \Omega^-$
    \FORALL{$(\omega^+, \omega^-) \in \Omega^+ \times \Omega^-$}
      \IF{$\mu(\proj_P(\omega^+)) = \mu(\proj_P(\omega^-))$}
        \STATE $\Pairsconf \gets \Pairsconf \cup \{(\omega^+, \omega^-)\}$
        \STATE Compute pivots $\Pi$ for $(\omega^+, \omega^-)$
        \STATE $\textsf{Mismatches}[(\omega^+, \omega^-)] \gets$ \\ \hspace*{\algorithmicindent} $\textsf{Mismatches}[(\omega^+, \omega^-)] \cup \Pi$
      \ENDIF
    \ENDFOR
  \ENDFOR
\ENDFOR
\RETURN $(\Pairsconf, \textsf{Mismatches})$
\end{algorithmic}
\end{algorithm}

The outer loop enumerates subsets of $\AP$ by decreasing cardinality;
intuitively, our aim is to identify trace pairs that become propositionally
indistinguishable once specific atomic propositions are discarded via
projection. For any pair $(\omega_1, \omega_2)$ that is propositionally
indistinguishable with respect to $P \subseteq \AP$, the pivot extractor
iterates over all ordered index pairs $1 \leq i < j \leq \ell$, emitting a
single pivot $((i, j), \lfloor \max\{\delay^1_{i,j},
\delay^2_{i,j}\} \rfloor)$ for each $(i, j)$.
We remark that in practice  it is often the case that $|\APorig| \leq 5$ for typical specifications, making
$2^{|\APorig|}$ a small constant. 
For example, the benchmarks in~\cite{wang2025synthesis} all have
$|\Sigma| < 5$, which corresponds to $|\APorig| < 3$.

\begin{proposition}[Phase~1 Complexity]\label{prop:phase1-complexity}
\Cref{alg:phase1} runs in time $\mathcal{O}\bigl(2^{|\APorig|} \cdot N^2 \cdot \ell^2\bigr)$
on a sample $(\Omega^+, \Omega^-)$ with $N = \max\{|\Omega^+|, |\Omega^-|\}$ traces of maximum length $\ell$.
\end{proposition}
\begin{proof}
The outer  enumeration contributes $2^{|\APorig|}$ subsets.
For each $P$, projecting all $N$ traces is $\mathcal{O}(N \cdot \ell)$. Comparing the
$\mathcal{O}(N^2)$ ordered cross-pairs costs $\mathcal{O}(\ell)$ per pair.
Pivot extraction per matched pair adds
$\mathcal{O}(\ell^2)$.
\end{proof}
\begin{example}\label{ex:threepairs}
Consider the sample $(\Omega^+, \Omega^-)$ with $\Omega^+ = \{ \omega^+_1, \omega^+_2, \omega^+_3\}$, $\Omega^- = \{ \omega^-_1, \omega^-_2, \omega^-_3\}$
where $\AP = \{p, q, p_1, q_1, p_2\}$:
\begin{align*}
  \omega^+_1 &= (1.5, \{p\}) (0.2, \{q_1\}) (1.0, \{p_2\}) (2.3, \{p_1\}) \\
  \omega^-_1 &= (1.5, \{p\}) (4.0, \{q\}) (4.1, \{q\}) (3.4, \{p_1\}) \\[1ex]
  \omega^+_2 &= (2.8, \{p\}) (3.2, \{q\}) (0.6, \{q_1\}) (0.6, \{p_1\}) \\
  \omega^-_2 &= (2.8, \{p\}) (1.5, \{p_2\}) (5.5, \{p_2\}) (1.0, \{p_1\}) \\[1ex]
  \omega^+_3 &= (4.1, \{p\}) (1.0, \{p_1\}) (2.0, \{p_2\}) (0.9, \{q\}) \\
  \omega^-_3 &= (4.1, \{p\}) (2.0, \{q_1\}) (4.0, \{p_2\}) (0.9, \{q\})
\end{align*}

Phase~1 enumerates subsets of $\AP$ in decreasing order of cardinality ($r = 5, 4, \dots, k$). Suppose that we set $k = 2$.

\begin{itemize}[leftmargin=*]
  \item At $r=5$ and $r=4$: no collisions occur.
  \item At $r=3$: 
    \begin{itemize}[leftmargin=*]
      \item $P=\{p, p_2, q\}$: The pair $(\omega^+_3, \omega^-_3)$ collides, both projecting to $\{p\}\,\emptyset\,\{p_2\}\,\{q\}$. There are more than one pivot for $(\proj_P(\omega_3^+), \proj_P(\omega_3^-))$, but as an example, the accumulated delays between $\{p\}$ (position 1) and $\{p_2\}$ (position 3) are $3.0$ and $6.0$ respectively, yielding $((1,3), 6)$.
      \item $P=\{p, q_1, p_2\}$: The pair $(\omega^+_1, \omega^-_3)$ collides, both projecting to $\{p\}\,\{q_1\}\,\{p_2\}\,\emptyset$.
      \item $P=\{p, p_1, p_2\}$: The pair $(\omega^+_2, \omega^-_1)$ collides, both projecting to $\{p\}\,\emptyset\,\emptyset\,\{p_1\}$.
    \end{itemize}
  \item At $r=2$ and $P=\{p, p_1\}$: 
    $(\omega^+_1, \omega^-_1)$, $(\omega^+_2, \omega^-_2)$ and $(\omega^+_1, \omega^-_2)$ both strip down to
    $\{p\}\,\emptyset\,\emptyset\,\{p_1\}$. 
\end{itemize}
\end{example}

\subsection{Phase~2: Template Injection}
\label{sec:phase2}
\Cref{alg:phase2} takes as input the set of conflicting trace pairs $\Pairsconf$
generated in Phase~1, alongside an optional structural template
$\phi_{\textit{template}}$. It outputs $\Pairsundone$, a reduced subset of trace
pairs that remain unresolved. 

The motivation for this phase is that, in many practical settings, domain
experts possess prior knowledge of specific behavioural patterns the synthesised
specification should respect---for example, `\emph{every arrival is eventually
followed by a completion}'. When provided, this template
$\phi_{\textit{template}}$ is efficiently evaluated at every position across all
traces using a sliding-window dynamic programming algorithm
(\Cref{prop:sliding-window-dp}). The resulting Boolean truth vectors are
then injected into all the traces as a \emph{free} new proposition.
Any conflicting pairs in $\Pairsconf$ that are successfully
distinguished by this newly injected template are immediately resolved and
discarded. The remaining, indistinguishable trace pairs form the output set
$\Pairsundone$. In this capacity, the template functions as a powerful
\emph{structural pre-filter}. It cleanly separates broad structural
violations---which the domain-knowledge template readily catches---from pure
quantitative timing violations, leaving only the latter for Phase~3 to address.

\begin{algorithm}[t]
\caption{Phase 2: Template Injection}
\label{alg:phase2}
\begin{algorithmic}[1]
\REQUIRE $\Pairsconf$; optional template $\phi_{\textit{template}}$
\ENSURE Updated $\AP$, $\Pairsundone$
\STATE $\AP \gets \APorig$;\quad $\Pairsundone \gets \Pairsconf$
\IF{$\phi_{\textit{template}}$ is provided}
  \STATE Evaluate $\phi_{\textit{template}}$ on every position of every trace
         via sliding-window DP
  \STATE $\AP \gets \AP \cup \{\phi_{\textit{template}}\}$
  \FORALL{$(\omega^+, \omega^-) \in \Pairsundone$}
    \IF{$\phi_{\textit{template}}$ at some position distinguishes $\omega^+$ from $\omega^-$}
      \STATE $\Pairsundone \gets \Pairsundone \setminus \{(\omega^+, \omega^-)\}$
    \ENDIF
  \ENDFOR
\ENDIF
\RETURN $(\AP, \Pairsundone)$
\end{algorithmic}
\end{algorithm}

\begin{proposition}[Sliding-Window DP]\label{prop:sliding-window-dp}
For any feature formula $\phi$ of the form
$\eventually_{[0,c]} \phi'$, $\once_{[0,c]} \phi'$, 
$\PnF_{[0,c)}(\phi'_1,\dots,\phi'_k)$,
or
$\PnO_{[0,c)}(\phi'_1,\dots,\phi'_k)$
where $\phi', \phi_1', \dots, \phi_k'$ are Boolean combinations of atomic propositions  and a timed word $\omega$ of length $\ell$, the Boolean satisfaction values
$\{(\omega, i) \models \phi : 1 \leq i \leq \ell\}$ can be computed in
$\mathcal{O}(\ell \cdot k)$ time and $\mathcal{O}(k)$ auxiliary space.
\end{proposition}

\begin{proof}[Proof sketch]
The evaluation of the standard unary $\mitl{}$ modalities $\eventually_{[0,c]} \phi'$ and $\once_{[0,c]} \phi'$ can be computed in $\mathcal{O}(\ell)$ time using standard sliding-window techniques (see, e.g.,~\cite{ho2014online}). We therefore focus on constructively demonstrating the evaluation algorithm for the Pnueli modalities.

For the $k$-ary future Pnueli modality $\PnF_{[0,c)}(\phi'_1, \dots, \phi'_k)$, we must evaluate
whether there exists a sequence of indices $j_1, j_2, \dots, j_k$ such that $i < j_1 < j_2 < \dots <
j_k$, $\Sigma_{m \in (i, j_k]} d_m < c$, and $\sigma_{j_m} \models \phi'_m$ for all $1 \leq m \leq
k$ (the case of the past Pnueli modality $\PnO$ is symmetric).
 We maintain a state array $M[1 \dots k]$, where $M[m]$ records the \emph{minimum ending position} of a valid matching subsequence $\phi'_m, \dots, \phi'_k$ that lie after the current position. We initialise $M[m] = \infty$ for all $1 \leq m \leq k$.

We iterate $i$ backwards from $\ell$ down to $1$. At each $i$:
\begin{enumerate}
    \item \textbf{Update the DP state:} We process the position $i+1$:
    \begin{itemize}
        \item For $m$ from $1$ up to $k-1$: if $\sigma_{i+1} \models \phi'_m$, we update $M[m] \gets \min(M[m], M[m+1])$.
        \item For $m=k$: if $\sigma_{i+1} \models \phi'_k$, we update $M[k] \gets \min(M[k], i+1)$.
    \end{itemize}
    \item \textbf{Update right boundary:} Update the right window boundary $R(i)$ leftwards to enforce that $\Sigma_{m \in (i, R(i)]} d_m < c$.
\end{enumerate}
The Pnueli modality holds at $i$ if and only if $M[1] \leq R(i)$.

The state array $M$ requires $\mathcal{O}(k)$ auxiliary space. For each of the $\ell$ steps, advancing the DP state takes $\mathcal{O}(k)$ operations. The right pointer $R(i)$ retreats at most $\ell$ times globally, amortising to $\mathcal{O}(1)$ per step. Thus, the overall evaluation time is bounded by $\mathcal{O}(\ell \cdot k)$.
\end{proof}

\begin{corollary}[Membership for $\mitppl{}$ formula]\label{thm:membership}
Given a timed word $\omega$ of length $\ell$, an $\mitppl{}$ formula $\varphi$ 
with $m$ temporal operators and maximum Pnueli arity $k$, the satisfaction relation
$(\omega,i) \models \varphi$ can be decided for all $i$ in time
$\mathcal{O}(\ell \cdot m \cdot k)$ and space $\mathcal{O}(m \cdot (\ell + k))$.
\end{corollary}

\subsection{Phase 3: Layered Feature Synthesis as Set Cover}
\label{sec:phase3}

\Cref{alg:phase3} iteratively selects a sequence of synthesised features $\phi_1, \phi_2,
\dots$ such that each $\phi_i$ distinguishes at least one pair in $\Pairsundone$
that has not been previously resolved by $\phi_1, \dots, \phi_{i-1}$. This
process continues until $\Pairsundone = \emptyset$. We formulate this as a
classical set-cover problem: the universe of elements to cover is the set of
unresolved pairs $\Pairsundone$, and the available sets are defined by the
distinguishing power of each candidate feature, denoted $\mathsf{Dist}(\phi) =
\{(\omega^+, \omega^-) \in \Pairsundone \mid \phi \text{ distinguishes }
(\omega^+, \omega^-)\}$. The objective is to cover the universe using the
minimum number of features. 

Given the intractability of optimal set cover, we solve this problem greedily. Notably, our synthesis approach is \emph{layered} (or compositional). When synthesising the $i$-th feature $\phi_i$, the framework may utilise not only the original atomic propositions ($\AP_{\mathrm{orig}}$) but also any of the previously synthesised and injected features $\{\phi_1, \dots, \phi_{i-1}\}$. This layered expansion allows the algorithm to express complex, nested timing constraints incrementally (e.g., placing a bounded deadline on a condition that itself contains a previously resolved temporal obligation) without suffering the combinatorial explosion of attempting to synthesise deeply nested $\mitppl{}$ formulae in a single pass.

\paragraph{Standard $\mitl{}$ features}\label{sec:phase3a}
For each unresolved pair $(\omega^+, \omega^-) \in \Pairsundone$ and each corresponding pivot $((i,j), c) \in \textsf{Mismatches}[(\omega^+, \omega^-)]$, we construct the candidate features:
\[
\phi_{(i,j),c}^{\eventually} = \F_{[0,c]}\, \sigma'_j \qquad \text{and} \qquad
\phi_{(i,j),c}^{\once} = \once_{[0,c]}\, \sigma'_i.
\]
These candidate features are subsequently evaluated at every position across all traces. (Note that we use the strict interval $[0, c)$ rather than $[0, c]$ if the larger delay $\delay_{(i, j)}$ is exactly equal to $c$). Here, $\sigma'_i$ and $\sigma'_j$ are interpreted as the Boolean conjunctions corresponding to $\sigma'_i, \sigma'_j \in \Sigma_{\AP}$, respectively.

\paragraph{Pnueli features}\label{sec:phase3b}
If the previous step yields an empty candidate set, or if every candidate has an empty $\mathsf{M}_\textsf{undone}$, it remains possible to distinguish the two timed words via \emph{counting}. For each $((i,j), c) \in \textsf{Mismatches}[(\omega^+, \omega^-)]$, we extract the subwords $w^+$ and $w^-$ around the pivot : specifically, the sequences of events from $\omega^+$ and $\omega^-$ that fall within the $[0, c]$ time window relative to position $i$. 

We then identify a shortest sequence $\sigma^*_1 \cdots \sigma^*_k$ that is a \emph{subsequence} of one bounded sub-word but not of the other
(this must exist since $w^+$ and $w^-$ are of different lengths). This constitutes a classical shortest distinguishing subsequence (SDS) problem~\cite{baeza1991searching}, which can be solved via a breadth-first search (BFS) in $\mathcal{O}(|w^+| \cdot |w^-| \cdot |\Sigma_{\AP}|)$ time. The extracted SDS $\sigma^*_1 \cdots \sigma^*_k$ naturally yields the Pnueli feature
$$\phi_{(i,j),c}^{\textit{Pnueli}} = \PnF_{[0,c]}(\sigma^*_1, \dots, \sigma^*_k)$$
By construction, because the SDS is present in one sub-word and absent from the other, $\phi_{(i,j),c}^{\textit{Pnueli}}$ is guaranteed to distinguish $(\omega^+, \omega^-)$ at position $i$.

\paragraph{Expressiveness} 
It is important to note that our feature synthesis phase exclusively generates
$\mitl{}$ 
and Pnueli modalities with bounded intervals of the form $[0, u \rangle$. This design choice, however, imposes no limitation on the theoretical
expressiveness of our framework. As established in~\cite{ho2026MightyPPL}, \emph{unilateral}
$\mitl{}$ and
Pnueli modalities are expressively complete for full $\mitppl{}$; for instance, a general bounded-until formula
such as $\varphi_1 \until_{[a,b]} \varphi_2$ can always be logically decomposed into an equivalent formula
using only unilateral intervals. Consequently, our approach requires no user-supplied templates
to achieve full expressiveness. Instead, we rely on the downstream untimed $\ltl{}$ learner to
compose these simple unilateral features into the logical equivalents of more complex modalities,
should they be necessary to separate the sample.

\paragraph{Global scoring}\label{sec:scoring}\label{sec:phase3c}

Relying exclusively on a feature's ability to resolve currently outstanding conflicting pairs can lead to severe overfitting, as the algorithm risks selecting highly specific, brittle formulae tailored to noisy edge cases. To promote structural generalisation, candidate features are instead evaluated by their broad distinguishing power across the entire dataset, rather than being strictly limited to pairs currently in $\Pairsundone$. 
Consequently, for the purpose of global scoring, we broaden the definition of a feature's distinguishing set $\mathsf{Dist}(\phi)$ to encompass the entire universal dataset:
\[
\mathsf{Dist}(\phi) =
\{(\omega^+, \omega^-) \in \Omega^+ \times \Omega^- \mid
\phi \text{ distinguishes } (\omega^+, \omega^-)\}
\]
This global set is computed using a single linear-time sliding-window dynamic programming pass over all traces (\Cref{prop:sliding-window-dp}).
Candidates are then scored based on the cardinality of their intersection with a target set, $|\mathsf{Dist}(\phi) \cap \Pairstgt|$. Here, $\Pairstgt \in \{\Pairsundone, \Pairsconf, \Omega^+ \times \Omega^-\}$ is a user-selectable parameter that dictates the \emph{generalisation reach} of the synthesis: targeting $\Pairsundone$ prioritises immediate local progress; targeting $\Pairsconf$ favours formulae that resolve a broader set of historical collisions; and targeting $\Omega^+ \times \Omega^-$ optimises generalisation to unseen but structurally similar pairs.
In the event of a tie, we prioritise candidates that maximise $|\P \cap \AP_{\mathrm{orig}}|$, where $\P$ is the set of atomic propositions occurring in $\phi$. We favour features heavily grounded in $\APorig$, as they are typically easier for downstream untimed $\ltl{}$ learners to process. Finally, any remaining ties are broken by favouring shorter formula lengths.

\paragraph{Progress guarantee} We select the first candidate $\phi^{*}$ in
sorted order whose $\mathsf{M}_\textsf{undone}(\phi^{*}) = \textsf{Dist}(\phi^{*}) \cap
\Pairsundone$ is non-empty. This selection rule ensures that
\[
|\Pairsundone^{(t+1)}| \;<\; |\Pairsundone^{(t)}|,
\]
i.e.~the undone-pair count strictly decreases on every iteration. Hence
the outer loop terminates in at most $|\Pairsconf|$ iterations.

\begin{algorithm}[t]
\caption{Phase 3: Layered Feature Synthesis}
\label{alg:phase3}
\begin{algorithmic}[1]
\REQUIRE $\Pairsundone$, $\textsf{Mismatches}$, $\AP$, $\Pairstgt$
\ENSURE Updated $\AP$ with synthesised features
\WHILE{$\Pairsundone \neq \emptyset$}
  \STATE $\mathsf{Candidates} \gets \emptyset$
  \FORALL{$(\omega^+, \omega^-) \in \Pairsundone$}
    \FORALL{pivot $((i,j),c) \in \textsf{Mismatches}[(\omega^+, \omega^-)]$}
      \FORALL{$\triangledown \in \{\eventually, \once\}$}
        \STATE $\phi \gets \triangledown_{[0,c]}\sigma'_j$ (or $\sigma'_i$ for $\once$)
        \STATE Compute $\textsf{Dist}(\phi)$, $\mathsf{M}_\textsf{undone}(\phi)$ via DP
        \STATE $\mathsf{Candidates} \gets \mathsf{Candidates} \cup \{(\phi\}$
      \ENDFOR
    \ENDFOR
  \ENDFOR
  \IF{$\mathsf{Candidates} = \emptyset$ \textbf{or} every candidate has empty $\mathsf{M}_\textsf{undone}$}
    \FORALL{$(\omega^+, \omega^-) \in \Pairsundone$}
      \STATE Extract $w^+, w^-$; find SDS $\sigma^*_1 \cdots \sigma^*_k$
      \STATE $\phi \gets \PnF_{[0,c)}(\sigma^*_1, \dots, \sigma^*_k)$
      \STATE Compute $\textsf{Dist}(\phi)$, $\mathsf{M}_\textsf{undone}(\phi)$ via DP
      \STATE $\mathsf{Candidates} \gets \mathsf{Candidates} \cup \{(\phi\}$
    \ENDFOR
  \ENDIF
  \STATE Sort $\mathit{Candidates}$ descending by $|\textsf{Dist} \cap \Pairstgt|$, then by $|P \cap \APorig|$, then ascending by length
  \STATE $\phi^* \gets$ first candidate in sorted order with $|\mathsf{M}_\textsf{undone}| > 0$
  \STATE $\AP \gets \AP \cup \{\phi^*\}$
  \STATE $\Pairsundone \gets \Pairsundone \setminus \mathsf{M}_\textsf{undone}(\phi^*)$
\ENDWHILE
\RETURN $\AP$
\end{algorithmic}
\end{algorithm}
\subsection{Phase 4: Proposition Minimisation}\label{sec:phase4}

While the expanded $\AP$ produced by Phase~3 successfully resolves all
structural ambiguities, it is not necessarily minimal. Because of the greedy
selection strategy, an earlier feature may become redundant if a subset of
subsequently added features collectively covers all the trace pairs it
originally distinguished, thereby rendering it obsolete. Phase~4 minimises $\AP$
by reducing this task to another classical set-cover instance, which is then
solved either exactly (via Petrick's method~\cite{petrick1956direct}) or approximately (via a greedy
hitting-set algorithm).
Specifically, let $\AP$ be the fully expanded set of atomic propositions output by Phase~3, and
let $\mathsf{U} = \Pairsconf$ be the universal set of originally conflicting
trace pairs. For each proposition $p \in \AP$, we define $\mathsf{Dist}(p)
\subseteq \mathsf{U}$ as the specific subset of conflicting pairs that $p$
successfully distinguishes. The objective is to identify a subset $\APfin
\subseteq \AP$ of minimum cardinality such that $\bigcup_{p \in \APfin}
\mathsf{Dist}(p) = \mathsf{U}$.

In practice, we observe that exact minimisation via Petrick's method is generally
tractable for the problem instances typically encountered, where the universe is
relatively small and the distinguishing sets are sparse. Upon completion, the
framework returns $(\Omega^+_{\textit{min}}, \Omega^-_{\textit{min}})$. These sets
are constructed by projecting the original traces exclusively onto $\APfin$ and
discarding all timing data. Because both sets now consist of strictly disjoint
untimed words over a minimised set of atomic propositions, an off-the-shelf
$\ltl{}$ learner can process them directly to mine the final specification.

\section{Theoretical Analysis}\label{sec:theory}

In this section, we formally establish the correctness and efficiency of our
proposed synthesis framework. First, we prove its \emph{soundness},
demonstrating that the layered feature synthesis restores propositional
distinguishability for all conflicting pairs, thereby validating our
reduction from timed synthesis to untimed $\ltl{}$ learning. Second, we
establish \emph{completeness}, ensuring that the algorithm is guaranteed to
terminate and return a valid separating set of atomic propositions. Finally, we analyse the \emph{computational complexity} of the
procedure.
In what follows, we assume that the sample 
$(\Omega^+, \Omega^-)$ is inherently distinguishable, i.e.~each trace pair
$(\omega^+, \omega^-) \in \Omega^+ \times \Omega^-$ is either propositionally or quantitatively
distinguishable.

\begin{lemma}[Termination]\label{thm:termination}
Let $\AP$ be the set of atomic propositions produced by Phase~3, and let $\APfin$ be the
output of Phase~4. For every $(\omega^+, \omega^-) \in \Pairsconf$, the
untimed projections $\mu(\proj_{\APfin}(\omega^+))$ and $\mu(\proj_{\APfin}(\omega^-))$ are
distinct untimed strings.
\end{lemma}

\begin{proof}
Every pair in
$\Pairsconf$ is removed either in Phase~2 (by the injection of template $\phi$) or, due to the progress guarantee of Phase~3, at some iteration in Phase~3 (by adding a feature
$\phi$ with $(\omega^+, \omega^-) \in \mathsf{Dist}(\phi)$).
In either case, $\phi$ evaluates to $\mathbf{true}$ at some position of $\omega^+$ and $\mathbf{false}$ at
the matching position of $\omega^-$ (or vice versa). Phase~4's set-cover constraint
forces $\phi$ (or some other feature that distinguishes $\omega^+, \omega^-$) into $\APfin$. Therefore
$\mu(\proj_{\APfin}(\omega^+)) \neq \mu(\proj_{\APfin}(\omega^-))$.
\end{proof}

\begin{corollary}
For a given sample $(\Omega^+, \Omega^-)$ over $\Sigma_{\APorig}$, the corresponding untimed sample
$(\mu(\proj_{\APfin}(\Omega^+)), \mu(\proj_{\APfin}(\Omega^-)))$ over $\Sigma_{\APfin}$ can be distinguished by an
$\ltl{}$ formula $\Phi$ over $\APfin$ at position $1$.
\end{corollary}

The following theorem is then immediate, by the definition of the semantics of $\mitppl{}$.

\begin{theorem}[Soundness]\label{thm:soundness}
The $\mitppl{}$ formula $\varphi = \Phi(\phi_1, \dots, \phi_r)$ over $\APorig$, where $\phi_1, \dots, \phi_r$ are the templates and features added in Phase~2 and Phase~3, distinguishes $(\Omega^+, \Omega^-)$ at position $1$.
\end{theorem}

The theorem below holds since 
in Phase~3, for each $((i,j), c) \in \textsf{Mismatches}[(\omega^+, \omega^-)]$,
the Pnueli feature $\phi_{(i,j),c}^{\textit{Pnueli}}$ is guaranteed to distinguish $(\omega^+, \omega^-)$ at position $i$.

\begin{theorem}[Completeness]\label{thm:completeness}
If an instance $(\Omega^+, \Omega^-)$ of \Cref{prob:timedlearning} has a solution, then applying Phases~1 to 4 also yields a solution $\varphi$.
\end{theorem}

\begin{corollary}
\label{cor:ta}
An instance $(\Omega^+, \Omega^-)$ of \Cref{prob:timedlearning} has a solution
if and only if $(\Omega^+, \Omega^-)$   
can be distinguished by a timed automaton.
\end{corollary}
\begin{proof}[Proof sketch]
By~\cref{thm:completeness}, if $(\Omega^+, \Omega^-)$
has a solution, there is an $\mitppl{}$ formula $\varphi$ that distinguishes $(\Omega^+, \Omega^-)$ at position $1$, and $\varphi$ can be translated into a language-equivalent timed automaton~\cite{ho2026MightyPPL}.
Conversely, if $(\Omega^+, \Omega^-)$ can be distinguished by a timed automaton, each 
trace pair
$(\omega^+, \omega^-) \in \Omega^+ \times \Omega^-$ must either be propositionally or quantitatively
distinguishable. By~\cref{thm:soundness}, we can apply Phases~1 to 4 and obtain a solution $\varphi$.
\end{proof}

Although timed automata are strictly more expressive than $\mitppl{}$ in general, Corollary \ref{cor:ta} shows that the gap disappears while distinguishing finite collections of timed words.

\begin{proposition}[Complexity]\label{thm:complexity}
The framework runs in time
$\mathcal{O}((N + 2^{|\AP|}) \cdot N^4 \cdot \ell^4)$
on a sample $(\Omega^+, \Omega^-)$ with $N = \max\{|\Omega^+|, |\Omega^-|\}$, $\ell$ the maximum length of traces, and $\AP$
the set of atomic propositions after Phase~3.
\end{proposition}

\begin{proof}
In each of the phases:
\begin{itemize}
\item Phase~1: $\mathcal{O}(2^{|\APorig|} \cdot N^2 \cdot \ell^2)$ by~\Cref{prop:phase1-complexity}.
\item Phase~2: $\mathcal{O}(N \cdot \ell)$ per template evaluation.
\item Phase~3: $|\Pairsconf|$ iterations in the worst case, each evaluating
$\mathcal{O}(|\Pairstgt| \cdot \ell^2)$ candidates by sliding-window DP at $\mathcal{O}((N + 2^{|\AP|}) \cdot \ell^2)$ per candidate, total $\mathcal{O}((N + 2^{|\AP|}) \cdot N^4 \cdot \ell^4)$.
\item Phase~4: $\mathcal{O}(|\AP|^{2} \cdot N^2)$ if we use the greedy hitting-set algorithm instead. \qedhere
\end{itemize}
\end{proof}

\section{Implementation and Experimental Evaluation}\label{sec:experiments}

We have implemented our synthesis framework, including subset-projection
collision detection, template injection, layered feature synthesis, and
proposition minimisation, in Python 3.14. To generate timed words for our
evaluations, we utilised the timed automata uniform sampler
\textsc{Wordgen}~\cite{barbot2023wordgen}, and when necessary, employed
$\mightyppl{}$~\cite{ho2026MightyPPL} to convert $\mitppl{}$ formulae into timed automata.
For the downstream untimed learning phase, we use the \textsc{Bolt}
tool~\cite{bathie2026ltl} as our $\ltl{}$ learning back-end. All experiments were
conducted on a desktop workstation with an Intel Core i9-13900K CPU and 64\,GB
of RAM.
While direct empirical comparisons are inherently difficult due to
variations in problem instance formulations and targeted formalisms, a direct
hardware-matched benchmarking was not possible because the implementation
for~\cite{wang2025synthesis} is not publicly available. Consequently, we frame
our experimental results to explicitly demonstrate the standalone feasibility
and efficiency of our translation pipeline, discussing its performance in the context of
recent state-of-the-art methods~\cite{raha2023synthesizing, wang2025synthesis} whenever possible.

To demonstrate the effectiveness and efficiency of our approach, we conducted
three sets of experiments. First, we evaluate the toolchain's performance on learning the behaviours of a simple timed automaton. Second, we
demonstrate the scalability of our approach when learning standard CPS
specification patterns, specifically evaluating the impact of domain-knowledge
template injection. Finally, we apply our synthesis approach to a case study
involving a train-gate controller to identify unknown timing constants.

\subsection{A Simple Timed Automaton}
\label{sec:simpleta}

In this subsection, we evaluate our framework's ability to capture the
behaviours of a simple timed automaton (\Cref{fig:simpleta}), adapted from~\cite[Section~VII.A]{wang2025synthesis}.
To precisely replicate the experimental setting
of~\cite[Section~VII.A]{wang2025synthesis}, positive traces of varying lengths
($\ell \in \{6, 7, 8, 9, 10\}$) are sampled directly from the automaton, with
the expected delay of each individual event set to $2$. Negative traces of
identical lengths are sampled, with the same expected delay, from a timed automaton accepting the complement timed
language (this automaton is easy to construct, despite timed automata not being closed
under complementation in the general case). We construct balanced trace sets
where $|\Omega_+| = |\Omega_-| = N$, for $N \in \{5, 7, 9, 11, 13, 15\}$. 
In~\Cref{tab:simpleta.noproj,tab:simpleta.proj}, each divided data cell contains two numbers:
the first denotes the execution time originally reported in~\cite{wang2025synthesis}, and the second represents the execution time of
our own toolchain. We emphasise that this juxtaposition is provided purely for
reference; it must not be treated as a direct, head-to-head comparison, given
the significant differences in the goal (their method finds minimal intersection- and renaming-free $\tre{}$s), underlying hardware, and running environments.

The experimental results demonstrate that our approach highly efficiently
synthesises separating $\mitppl{}$ formulae---which, for this specific
benchmark, resolve entirely to standard $\mitl{}$.
Whilst enabling the subset-projection strategy incurs a minor runtime overhead, it yields
slightly more concise specifications: the average number of operators in the
resulting formulae is $8.63$, compared to $9.90$ when subset projection is
disabled. 
Conversely, if we strictly minimise the output set of propositions, the downstream $\ltl{}$ learner must compensate by generating more sophisticated structural constraints: the average operator counts increase to $10.73$ (with
subset projection) and $10.23$ (without), highlighting the delicate balance
between proposition size and structural formula complexity.

\begin{figure}[!htbp]
    \centering
\begin{tikzpicture}[auto, ->, node distance=5cm, transform shape, scale=0.7]

\node[state, initial] (q0) {$q_0$};
\node[state] (q1) [below left=3cm and 4.5cm of q0] {$q_1$};
\node[state, accepting] (q2) [below right=3cm and 4.5cm of q0] {$q_2$};

\path[->]
    (q0) edge [loopabove] node {$\{p\}, x < 2, \{x\}$} (q0)
    
    (q0) edge [bend right=15] node [left=0.5cm] {$\emptyset, x < 4, \{x\}$} (q1)
    (q1) edge  node [near start, right=0.5cm] {$\emptyset, x < 4, \{x\}$} (q0)
    
    (q0) edge  node [near end, left=0.5cm] {$\{p\}, x \geq 2 \land x < 4, \{x\}$} (q2)
    (q2) edge [bend right=15] node [right=0.5cm] {$\{p\}, x \geq 2 \land x < 4, \{x\}$} (q0)
    (q2) edge node [above] {$\{p\}, x < 2, \{x\}$} (q1)
    (q1) edge [bend right=15] node [below] {$\{p\}, x < 4, \{x\}$} (q2)
    
    (q2) edge [loop right, looseness=6, in=285, out=345] node[below] {$\emptyset, x < 4, \{x\}$} (q2);

\end{tikzpicture}
    \captionof{figure}{The simple timed automaton used in~\Cref{sec:simpleta}. Note that the alphabet size is $2$.}
    \label{fig:simpleta}
\end{figure}
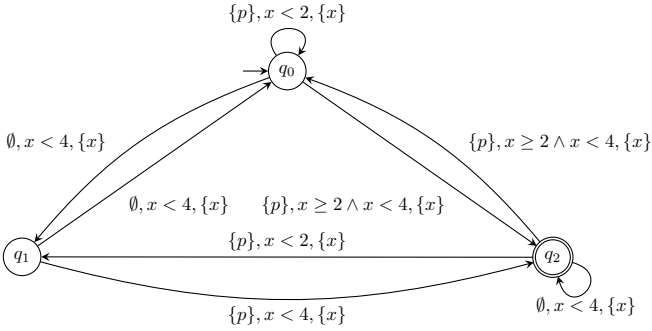

\begin{table}[!htbp]
\captionof{table}{Execution times (s) on the simple timed automaton benchmark (without subset projection). In each cell the first number is the reported execution time in~\cite{wang2025synthesis}, the second number is the execution time of our framework.}
\label{tab:simpleta.noproj}
\centering
\scalebox{0.75}{
\def\arraystretch{1.5}
\setlength\tabcolsep{0.5mm}
\begin{tabularx}{0.6\textwidth}{|c|*{12}{>{\raggedleft\arraybackslash}X|}}
\hline
\diagbox[width=3em, height=1.5\line]{$\ell$}{$N$} & 
\multicolumn{2}{c|}{5} & 
\multicolumn{2}{c|}{7} & 
\multicolumn{2}{c|}{9} & 
\multicolumn{2}{c|}{11} & 
\multicolumn{2}{c|}{13} & 
\multicolumn{2}{c|}{15} \\
\hline
6     & 0.41  & \textbf{0.04}    & 107.8  & \textbf{0.04}   & 0.33  & \textbf{0.08}     & 245.9 & \textbf{0.13}    & 125.8 & \textbf{0.18}  & 122.5 & \textbf{0.18} \\ \hline
7     & 32.6  & \textbf{0.02}    & 128.5  & \textbf{0.04}   & 125.7 & \textbf{0.08}     & 120.1 & \textbf{0.10}    & 115.0 & \textbf{0.19}  & 120.6 & \textbf{0.47} \\ \hline
8     & 1.52  & \textbf{0.04}    & 32.1   & \textbf{0.10}   & 154.7 & \textbf{0.05}     & 119.2 & \textbf{0.05}    & 124.9 & \textbf{0.18}  & 118.8 & \textbf{0.15} \\ \hline
9     & 122.5 & \textbf{0.04}    & 1.48   & \textbf{0.04}   & 124.3 & \textbf{0.09}     & 122.4 & \textbf{0.14}    & 117.5 & \textbf{0.17}  & 120.1 & \textbf{0.16} \\ \hline
10    & 112.7 & \textbf{0.04}    & 127.0  & \textbf{0.04}   & 127.2 & \textbf{0.05}     & 124.8 & \textbf{0.06}    & 116.4 & \textbf{0.20}  & 123.2 & \textbf{0.12} \\ \hline
\end{tabularx}
}
\end{table}

\begin{table}[!htbp]
\captionof{table}{Execution times (s) on the simple timed automaton benchmark (with subset projection). In each cell the first number is the reported execution time in~\cite{wang2025synthesis}, the second number is the execution time of our framework.}
\label{tab:simpleta.proj}
\centering
\scalebox{0.75}{
\def\arraystretch{1.5}
\setlength\tabcolsep{0.5mm}
\begin{tabularx}{0.6\textwidth}{|c|*{12}{>{\raggedleft\arraybackslash}X|}}
\hline
\diagbox[width=3em, height=1.5\line]{$\ell$}{$N$} & 
\multicolumn{2}{c|}{5} & 
\multicolumn{2}{c|}{7} & 
\multicolumn{2}{c|}{9} & 
\multicolumn{2}{c|}{11} & 
\multicolumn{2}{c|}{13} & 
\multicolumn{2}{c|}{15} \\
\hline
6     & 0.41  & \textbf{0.08}    & 107.8  & \textbf{0.11}   & 0.33  & \textbf{0.29}     & 245.9 & \textbf{0.25}    & 125.8 & \textbf{0.61}  & 122.5 & \textbf{0.70} \\ \hline
7     & 32.6  & \textbf{0.07}    & 128.5  & \textbf{0.17}   & 125.7 & \textbf{0.20}     & 120.1 & \textbf{0.24}    & 115.0 & \textbf{0.62}  & 120.6 & \textbf{0.58} \\ \hline
8     & 1.52  & \textbf{0.07}    & 32.1   & \textbf{0.18}   & 154.7 & \textbf{0.23}     & 119.2 & \textbf{0.24}    & 124.9 & \textbf{0.60}  & 118.8 & \textbf{0.46} \\ \hline
9     & 122.5 & \textbf{0.13}    & 1.48   & \textbf{0.17}   & 124.3 & \textbf{0.24}     & 122.4 & \textbf{0.31}    & 117.5 & \textbf{0.41}  & 120.1 & \textbf{0.58} \\ \hline
10    & 112.7 & \textbf{0.12}    & 127.0  & \textbf{0.22}   & 127.2 & \textbf{0.28}     & 124.8 & \textbf{0.35}    & 116.4 & \textbf{0.40}  & 123.2 & \textbf{0.64} \\ \hline
\end{tabularx}
}
\end{table}
\subsection{Standard Patterns in CPS}
\label{sec:patterns}

In this subsection, we evaluate the scalability of our method against standard
specification patterns common in real-time task scheduling and the online
monitoring of cyber-physical systems~\cite{raha2023synthesizing}. These include
structural archetypes such as bounded response ($\globally(p \implies
\eventually_{[0, 2]} q)$) and delayed acknowledgement. 

To this end, we fix the sample size to $N = 40$ and select four representative
$\mitl{}$ formulae from~\cite[Section~6]{raha2023synthesizing}: \[
\begin{aligned} & \globally(\eventually_{[0, 2]} p),  & \globally(\neg p
\implies \eventually_{[0, 2]} q) \\ & \globally(\neg p \implies \globally_{[0,
2]} q),  & \globally(p \until_{[0, 2]} q) \;. \end{aligned} \] For each target
formula $\varphi$, positive and negative traces of varying lengths ($\ell \in
\{5, 10, 15, 20\}$) are sampled from the language-equivalent timed automata for
$\varphi$ and $\neg \varphi$, respectively, constructed via $\mightyppl{}$.
In~\Cref{tab:pattern.noproj,tab:pattern.proj}, each divided data cell reports
two execution times to highlight the impact of domain knowledge: the first
denotes the baseline synthesis time without template injection, and the second
denotes the time when a corresponding structural template ($\eventually_{[0,
2]}p$, $\eventually_{[0, 2]} q$,  $\globally_{[0, 2]} q$, and $p \until_{[0, 2]}
q$, respectively) is injected before the main synthesis loop.

The experimental results confirm the efficiency of our toolchain in
resolving this benchmark suite. For context, the direct synthesis times reported
for these patterns in~\cite[Section~6]{raha2023synthesizing} range from $10$ to
$5400$ seconds. We reiterate, however, that these figures are provided purely as
a reference point rather than for direct empirical comparison, as~\cite{raha2023synthesizing} solves a related but different problem (find minimal formulae with bounded future-reach).
We also
observe that enabling the subset-projection strategy yields significantly more
concise specifications, reducing the average number of operators from $11.44$
down to $8.88$. Finally, we evaluate the effect of template injection. We emphasise
that templates represent a form of domain knowledge or an `oracle' that is not
necessarily available in all practical settings. However, in scenarios where such
structural priors are known and injected, they function as pre-computed, highly
selective features. The downstream \textsc{Bolt} learner immediately incorporates
these injected features, effectively bypassing the need to synthesise complex,
quantitative structural constraints from scratch.

\begin{table}[!htbp]
\captionof{table}{Execution times (s) on the CPS patterns benchmark (without subset projection).
In each cell the first number is the execution time without template injection, the second number is the execution time with template injection.}
\label{tab:pattern.noproj}
\centering
\scalebox{0.75}{
\def\arraystretch{1.5}
\setlength\tabcolsep{0.5mm}
\begin{tabularx}{0.6\textwidth}{|c|*{8}{>{\raggedleft\arraybackslash}X|}}
\hline
\diagbox[width=3em, height=1.5\line]{$\ell$}{$\phi$} & 
\multicolumn{2}{c|}{$\globally(\eventually_{[0, 2]} p)$} & 
\multicolumn{2}{c|}{$\globally(\neg p \Rightarrow \eventually_{[0, 2]} q)$} & 
\multicolumn{2}{c|}{$\globally(\neg p \Rightarrow \globally_{[0, 2]} q)$} & 
\multicolumn{2}{c|}{$\globally(p \until_{[0, 2]} q)$} \\
\hline
5    & 0.12  & \textbf{0.07}    & 0.43  & \textbf{0.07}   & 1.54          & \textbf{0.07}   & \textbf{0.07}  & \textbf{0.07}  \\ \hline
10   & 0.17  & \textbf{0.12}    & 2.58  & \textbf{0.12}   & 0.43          & \textbf{0.12}   & 0.43           & \textbf{0.12}  \\ \hline
15   & 2.06  & \textbf{0.12}    & 2.05  & \textbf{0.12}   & \textbf{0.12} & \textbf{0.12}   & \textbf{0.12}  & \textbf{0.12}  \\ \hline
20   & 0.64  & \textbf{0.12}    & 2.82  & \textbf{0.17}   & \textbf{0.12} & 0.17            & \textbf{0.12}  & 0.17  \\ \hline
\end{tabularx}
}
\end{table}

\begin{table}[!htbp]
\captionof{table}{Execution times (s) on the CPS patterns benchmark (with subset projection).
In each cell the first number is the execution time without template injection, the second number is the execution time with template injection.}
\label{tab:pattern.proj}
\centering
\scalebox{0.75}{
\def\arraystretch{1.5}
\setlength\tabcolsep{0.5mm}
\begin{tabularx}{0.6\textwidth}{|c|*{8}{>{\raggedleft\arraybackslash}X|}}
\hline
\diagbox[width=3em, height=1.5\line]{$\ell$}{$\phi$} &
\multicolumn{2}{c|}{$\globally(\eventually_{[0, 2]} p)$} &
\multicolumn{2}{c|}{$\globally(\neg p \Rightarrow \eventually_{[0, 2]} q)$} &
\multicolumn{2}{c|}{$\globally(\neg p \Rightarrow \globally_{[0, 2]} q)$} &
\multicolumn{2}{c|}{$\globally(p \until_{[0, 2]} q)$} \\
\hline
5  & 0.42 & \textbf{0.33} & 0.68 & \textbf{0.29} & 0.60 & \textbf{0.40} & 0.57 & \textbf{0.27} \\ \hline
10 & 1.57 & \textbf{0.77} & 8.40 & \textbf{0.98} & 2.34 & \textbf{0.93} & 6.32 & \textbf{0.88} \\ \hline
15 & 4.30 & \textbf{1.67} & 9.76 & \textbf{1.93} & 5.48 & \textbf{1.92} & 6.70 & \textbf{1.67} \\ \hline
20 & 8.89 & \textbf{3.57} & 17.49 & \textbf{3.63} & 11.11 & \textbf{3.62} & 10.54 & \textbf{4.68} \\ \hline
\end{tabularx}
}
\end{table}

\subsection{Case Study: Train-Gate Controller}
\label{sec:traingate}

Our final case study involves extracting timing constants from a train-gate
controller, modelled by a timed automaton and adapted
from~\cite[Section~VII.C]{wang2025synthesis}. For this scenario, we assume a
grey-box setting where the general `shape' or logical sequence of the system's
behaviour is known to the practitioner, but the precise integer time boundaries
must be learned entirely from the sampled traces.
Positive traces are sampled directly from the timed automaton shown
in~\Cref{fig:traingate}, with the expected delay of each individual event set to
$10$. Negative traces are generated using the same expected delay, but are
sampled from a structurally identical automaton where all guard conditions have
been relaxed to $\mathbf{true}$. The samples comprise traces of varying lengths;
for instance, the rows marked `$\ell \leq 10$' correspond to trace lengths $\ell
\in \{4, 5, 8, 9, 10\}$, which are the only valid lengths up to $10$ accepted by
the original automaton. As before, we construct balanced trace sets where
$|\Omega_+| = |\Omega_-| = N$.

\paragraph{The effect of proposition minimisation}
For this part of experiment we consider $N \in \{25, 50, 75, 100\}$.
In~\Cref{tab:traingate.noproj,tab:traingate.proj}, each divided data cell
contains two numbers: the first denotes the execution time without proposition
minimisation, and the second denotes the time with minimisation enabled.

Our framework performs efficiently on the moderate sample sizes evaluated here.
Consistent with previous observations, enabling the subset-projection strategy
slightly reduces the structural complexity of the output formulae, dropping the
average number of operators from $37.25$ to $36.83$ (when proposition
minimisation is disabled). Furthermore, whilst proposition minimisation can
yield up to a $2\times$ performance boost in certain cases, this speed-up comes
at the cost of significantly more complex formulae: the average operator counts
surge to $64.67$ (without subset-projection) and $63.17$ (with
subset-projection).

\paragraph{Scalability bottleneck and phase-by-phase breakdown}
To identify the primary computational bottleneck, \Cref{tab:traingate.breakdown} provides a runtime
breakdown on extended sample sizes ($N \in \{25, \dots, 750\}$, as
in~\cite[Section~VII.C]{wang2025synthesis}). Our feature synthesis (Phases 1--3) and proposition
minimisation (Phase 4) remain highly efficient, completing in roughly 9 minutes even for the largest sets of 750 traces.
The primary computational limitation, however, arises during the downstream \textsc{Bolt} execution. Because \textsc{Bolt}
performs an exact structural search over an exponentially large space, it times out on datasets of 300
or more traces when proposition minimisation is disabled. However, enabling minimisation drastically
compresses this search space, allowing \textsc{Bolt} to scale up to 600 traces.
This confirms that our frontend timing reduction is fundamentally scalable; as downstream untimed learners improve,
the overall capacity of our framework will naturally increase.

\paragraph{Generalisation and formula quality}
Because \textsc{Bolt} cannot scale beyond $N=100$ without enabling 
proposition minimisation---which acts as a form of lossy compression on the 
logical search space---we disable minimisation for the feature-level analysis 
of this experiment. \Cref{tab:traingate.features} lists the raw timing features 
extracted (Phases 1--3) for larger trace sets ($N \geq 100$). For shorter traces 
($\ell \leq 10$), the framework consistently isolates bounds matching the automaton's 
exact internal guards (e.g., $2$, $3$, $7$, and $12$), converging tightly on these 
physical bounds as $N$ increases. For larger bounds ($\ell \in \{15, 19\}$), the 
features incorporate larger constants (e.g., $24$, $36$, or $62$). Because longer 
traces encompass multiple structural loops, these expanded intervals successfully 
capture macro-level systemic delays that reflect the accumulated sums of base 
constants across repeated iterations. Examples of the corresponding full $\mitppl{}$ 
formulae synthesised from these features are provided in~\Cref{tab:traingate.rawformulae}.
Furthermore, to assess generalisation, we evaluated these formulae against a single 
large hold-out set comprising $10,000$ unseen test traces ($5,000$ positive, $5,000$ 
negative) generated up to $\ell \leq 19$. \Cref{tab:traingate.accuracy} illustrates 
the out-of-sample accuracy of specifications trained on varying sample sizes ($N \leq 100$). 
The results highlight a clear structural trade-off: specifications learned \emph{without} 
proposition minimisation consistently achieve higher accuracy (reaching up to $94.48\%$). 
By retaining full feature granularity, the uncompressed alphabet allows the exact solver 
to better capture the ground truth. In contrast, while proposition minimisation is strictly 
necessary to prevent the downstream \textsc{Bolt} search from timing out on larger datasets, 
this logical compression artificially restricts the solver's vocabulary, forcing a measurable 
degradation in accuracy on unseen traces.

\begin{figure}[!htbp]
    \centering

\begin{tikzpicture}[auto, ->, node distance=5cm, transform shape, scale=0.7]

\node[state, initial, accepting] (q0) {$q_0$};
\node[state] (q1) [right=3cm of q0] {$q_1$};
\node[state] (q2) [below=2cm of q1] {$q_2$};
\node[state] (q4) [below=2cm of q0] {$q_4$};
\node[state] (q3) at ($(q4)!0.5!(q2) + (0,-2.2cm)$) {$q_3$};

\path[->]
    (q0) edge node [above] {$\{p\}, \mathbf{true}, \{x\}$} (q1)
    (q1) edge node [right] {$\{q\}, x \geq 2, \{x\}$} (q2)
    (q2) edge node [above] {$\{r\}, x \geq 12, \{x\}$} (q4)
    (q4) edge node [left] {$\emptyset, x \geq 3, \{x\}$} (q0)

    (q2) edge node [right=4pt] {$\{p, q\}, \mathbf{true}, \{x\}$} (q3)
    (q3) edge node [left=4pt] {$\{r\}, x \geq 7, \{x\}$} (q4);

\end{tikzpicture}

    \captionof{figure}{The train-gate timed automaton used in~\Cref{sec:traingate}. Note that
    the alphabet size is $5$.}
    \label{fig:traingate}
\end{figure}
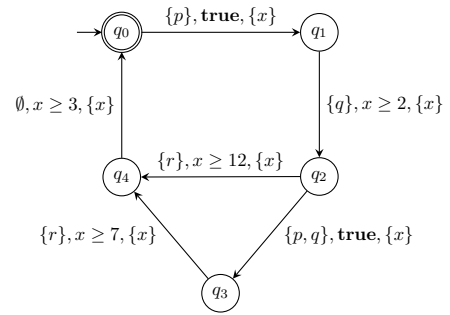

\begin{table}[!htbp]
\captionof{table}{Execution times (s) on the train-gate benchmark (without subset projection).
In each cell the first number is the execution time without proposition minimisation, the second number is the execution time with proposition minimisation.}
\label{tab:traingate.noproj}
\centering
\scalebox{0.75}{
\def\arraystretch{1.5}
\setlength\tabcolsep{0.5mm}
\begin{tabularx}{0.6\textwidth}{|c|*{8}{>{\raggedleft\arraybackslash}X|}}
\hline
\diagbox[width=3em, height=1.5\line]{$\ell$}{$N$} &
\multicolumn{2}{c|}{25} &
\multicolumn{2}{c|}{50} &
\multicolumn{2}{c|}{75} &
\multicolumn{2}{c|}{100} \\
\hline
$\leq 10$   & 1.77  & \textbf{0.67}    & 10.25  & \textbf{4.94}   & 32.33 & \textbf{15.46}  & 20.09 & \textbf{11.38}  \\ \hline
$\leq 15$   & 1.82  & \textbf{1.19}    & 11.19  & \textbf{8.39}   & 25.77 & \textbf{9.60}   & 37.29 & \textbf{16.04}  \\ \hline
$\leq 19$   & 1.42  & \textbf{0.56}    & \textbf{6.17}   & 8.50   & 11.11 & \textbf{7.70}   & 24.15 & \textbf{18.83}  \\ \hline
\end{tabularx}
}
\end{table}

\begin{table}[!htbp]
\captionof{table}{Execution times (s) on the train-gate benchmark (with subset projection).
In each cell the first number is the execution time without proposition minimisation, the second number is the execution time with proposition minimisation.}
\label{tab:traingate.proj}
\centering
\scalebox{0.75}{
\def\arraystretch{1.5}
\setlength\tabcolsep{0.5mm}
\begin{tabularx}{0.6\textwidth}{|c|*{8}{>{\raggedleft\arraybackslash}X|}}
\hline
\diagbox[width=3em, height=1.5\line]{$\ell$}{$N$} &
\multicolumn{2}{c|}{25} &
\multicolumn{2}{c|}{50} &
\multicolumn{2}{c|}{75} &
\multicolumn{2}{c|}{100} \\
\hline
$\leq 10$ & 3.49  & \textbf{1.50} & 10.05  & \textbf{6.22} & 41.33 & \textbf{19.48} & 39.11 & \textbf{15.99}  \\ \hline
$\leq 15$ & 3.00  & \textbf{2.23} & 14.00  & \textbf{8.80} & 30.66 & \textbf{14.52} & 40.00 & \textbf{21.08}  \\ \hline
$\leq 19$ & 3.70  & \textbf{2.78} & \textbf{10.75}  & 12.71 & 24.34 & \textbf{18.68}  & 32.95 & \textbf{28.42}  \\ \hline
\end{tabularx}
}
\end{table}

\begin{table}[!htbp]
\captionof{table}{Phase-by-phase execution time breakdown (s) for the train-gate benchmark ($\ell \leq
19$, with subset projection).}
\label{tab:traingate.breakdown}
\centering
\scalebox{0.75}{
\def\arraystretch{1.5}
\setlength\tabcolsep{0.5mm}
\begin{tabularx}{0.6\textwidth}{|c|*{5}{>{\raggedleft\arraybackslash}X|}}
\hline
\textbf{$N$} &
\multicolumn{1}{>{\centering\arraybackslash}X|}{Subset Projection} &
\multicolumn{1}{>{\centering\arraybackslash}X|}{Layered Feature Synthesis} &
\multicolumn{1}{>{\centering\arraybackslash}X|}{Proposition Minimisation} &
\multicolumn{1}{>{\centering\arraybackslash}X|}{\textsc{Bolt} w/o Proposition Minimisation} &
\multicolumn{1}{>{\centering\arraybackslash}X|}{\textsc{Bolt} w/ Proposition Minimisation} \\ \hline
25  & 0.05  & 2.35   & 0.00  & 3.34   & \textbf{2.56}   \\ \hline
50  & 0.21  & 9.82   & 0.00  & 9.08   & \textbf{14.46}  \\ \hline
75  & 0.53  & 12.22  & 0.01  & 25.13  & \textbf{16.02}  \\ \hline
100 & 0.91  & 20.04  & 0.02  & 33.12  & \textbf{27.51}  \\ \hline
300 & 8.74  & 100.56 & 0.24  & ERR    & \textbf{123.58} \\ \hline
450 & 21.82 & 196.79 & 0.71  & ERR    & \textbf{243.26} \\ \hline
600 & 38.14 & 316.19 & 1.76  & ERR    & \textbf{425.13} \\ \hline
750 & 60.87 & 478.81 & 3.40  & ERR    & ERR    \\ \hline
\end{tabularx}
}
\end{table}

\begin{table}[!htbp]
\captionof{table}{Timing features synthesised (without proposition minimisation).}
\label{tab:traingate.features}
\centering
\scalebox{0.75}{
\def\arraystretch{1.5}
\setlength\tabcolsep{1.5mm}
\begin{tabularx}{0.6\textwidth}{|c|c|>{\small\arraybackslash}X|}
\hline
\textbf{$\ell$} & \textbf{$N$} & \multicolumn{1}{c|}{\small Synthesised Features (Phases 1--3)} \\ \hline
\multirow{5}{*}{$\leq 10$}
 & 100 & $\eventually_{[0, 7]}(\mathbf{true})$, $\eventually_{[0, 18]}(\neg r)$, $\eventually_{[0, 40]}(\neg q \land \neg r)$ \\ \cline{2-3}
 & 300 & $\eventually_{[0, 8]}(\mathbf{true})$, $\eventually_{[0, 16]}(\neg r)$, $\eventually_{[0, 3]}(\mathbf{true})$, $\eventually_{[0, 12]}(\mathbf{true})$, $\eventually_{[0, 18]}(\neg q)$ \\ \cline{2-3}
 & 450 & $\eventually_{[0, 7]}(\mathbf{true})$, $\eventually_{[0, 18]}(\neg r)$, $\eventually_{[0, 2]}(\mathbf{true})$, $\eventually_{[0, 12]}(\neg p)$ \\ \cline{2-3}
 & 600 & $\eventually_{[0, 7]}(\mathbf{true})$, $\eventually_{[0, 18]}(\neg p \land \neg r)$, $\eventually_{[0, 12]}(\mathbf{true})$, $\eventually_{[0, 3]}(\mathbf{true})$, $\eventually_{[0, 2]}(\mathbf{true})$ \\ \cline{2-3}
 & 750 & $\eventually_{[0, 7]}(\mathbf{true})$, $\eventually_{[0, 17]}(\neg p \land \neg r)$, $\eventually_{[0, 2]}(\mathbf{true})$, $\eventually_{[0, 12]}(\mathbf{true})$, $\eventually_{[0, 3]}(\mathbf{true})$ \\ \hline
\multirow{5}{*}{$\leq 15$}
 & 100 & $\eventually_{[0, 10]}(\neg p)$, $\eventually_{[0, 2]}(\mathbf{true})$, $\eventually_{[0, 40]}(\neg q \land \neg r)$ \\ \cline{2-3}
 & 300 & $\eventually_{[0, 11]}(\neg p)$, $\eventually_{[0, 4]}(\mathbf{true})$, $\eventually_{[0, 22]}(\neg r)$, $\eventually_{[0, 12]}(\neg q)$ \\ \cline{2-3}
 & 450 & $\eventually_{[0, 7]}(\mathbf{true})$, $\eventually_{[0, 24]}(\neg p \land \neg r)$, $\eventually_{[0, 11]}(\mathbf{true})$, $\eventually_{[0, 2]}(\mathbf{true})$ \\ \cline{2-3}
 & 600 & $\eventually_{[0, 8]}(\mathbf{true})$, $\eventually_{[0, 34]}(\neg p \land \neg q \land \neg r)$, $\eventually_{[0, 4]}(\mathbf{true})$, $\eventually_{[0, 12]}(\mathbf{true})$ \\ \cline{2-3}
 & 750 & $\eventually_{[0, 8]}(\mathbf{true})$, $\eventually_{[0, 19]}(r)$, $\eventually_{[0, 4]}(\mathbf{true})$, $\eventually_{[0, 2]}(\mathbf{true})$, $\eventually_{[0, 16]}(\neg r)$, $\eventually_{[0, 29]}(\neg q \land \neg r)$, $\eventually_{[0, 62]}(\neg q \land \neg r)$ \\ \hline
\multirow{5}{*}{$\leq 19$}
 & 100 & $\eventually_{[0, 20]}(\neg p \land \neg r)$, $\eventually_{[0, 46]}(\neg q \land \neg r)$ \\ \cline{2-3}
 & 300 & $\eventually_{[0, 9]}(\neg p)$, $\eventually_{[0, 36]}(\neg q \land \neg r)$, $\eventually_{[0, 3]}(\mathbf{true})$, $\eventually_{[0, 12]}(\mathbf{true})$ \\ \cline{2-3}
 & 450 & $\eventually_{[0, 11]}(\neg r)$, $\eventually_{[0, 12]}(\neg p)$, $\eventually_{[0, 4]}(\mathbf{true})$ \\ \cline{2-3}
 & 600 & $\eventually_{[0, 8]}(\mathbf{true})$, $\eventually_{[0, 19]}(\neg p \land \neg r)$, $\eventually_{[0, 2]}(\mathbf{true})$, $\eventually_{[0, 12]}(\mathbf{true})$ \\ \cline{2-3}
 & 750 & $\eventually_{[0, 8]}(\mathbf{true})$, $\eventually_{[0, 20]}(\neg p \land \neg r)$, $\eventually_{[0, 3]}(\mathbf{true})$, $\eventually_{[0, 14]}(\neg \eventually_{[0, 20]}(\neg p \land \neg r))$, $\eventually_{[0, 11]}(\mathbf{true})$ \\ \hline
\end{tabularx}
}
\end{table}

\begin{table}[!htbp]
\captionof{table}{Example formulae synthesised ($\ell \leq 10$, without proposition minimisation).}
\label{tab:traingate.rawformulae}
\centering
\scalebox{0.75}{
\def\arraystretch{1.5}
\setlength\tabcolsep{1.5mm}
\begin{tabularx}{0.6\textwidth}{|c|c|>{\scriptsize\ttfamily\raggedright\arraybackslash}X|}
\hline
\textbf{$\ell$} & \textbf{$N$} &
\multicolumn{1}{c|}{\normalfont\small{Synthesised $\mitppl{}$ Formula}} \\ \hline
\multirow{4}{*}{\normalsize $\leq 10$}
 & \normalsize 25 & 
(G ((X[!] (X (X[!] (X[!] (r))))) -> (F[0, 29](!(q) \&\& !(r))))) \& ((X ((! (F[0, 29](!(q) \&\& !(r)))) R ((F[0, 42](!(q))) U (F[0, 29](!(q) \&\& !(r)))))) \textbar{} ((F[0, 42](!(q))) \& (G ((q) <-> ((q) R ((q) U (F[0, 29](!(q) \&\& !(r)))))))))
 \\ \cline{2-3}
 & \normalsize 50 & 
(G (F (! (r)))) \& ((G ((F[0, 14](!(r))) <-> (r))) \textbar{} ((! ((r) R ((F[0, 15](!(q))) U (! (F[0, 14](!(r))))))) \& ((G ((F[0, 15](!(q))) -> ((X (X (F[0, 15](!(q))))) -> (r)))) \textbar{} ((G ((X[!] (p)) -> (F[0, 14](!(r))))) \& ((! (F[0, 15](!(q)))) \& (G (((F[0, 14](!(r))) U (X[!] (X (false)))) -> (r))))))))
 \\ \hline
\end{tabularx}
}
\end{table}

\begin{table}[!htbp]
\captionof{table}{Out-of-sample generalisation accuracy evaluated on a held-out set of $10,000$
unseen traces ($\ell \leq 19$).}
\label{tab:traingate.accuracy}
\centering
\scalebox{0.75}{
\def\arraystretch{1.3}
\setlength\tabcolsep{1.5mm}
\begin{tabularx}{0.6\textwidth}{|c|c|>{\centering\arraybackslash}X|>{\centering\arraybackslash}X|}
\hline
\textbf{$\ell$} & \textbf{Training Size ($N$)} & \textbf{Accuracy (w/o Prop. Min.)} & \textbf{Accuracy (w/ Prop. Min.)} \\ \hline
\multirow{4}{*}{$\leq 10$} 
 & 25  & \textbf{77.22}\%   & 71.27\%           \\ \cline{2-4}
 & 50  & 78.64\%            & \textbf{80.12}\%  \\ \cline{2-4}
 & 75  & \textbf{83.63}\%   & 77.80\%           \\ \cline{2-4}
 & 100 & \textbf{91.97}\%   & 89.65\%  \\ \hline
\multirow{4}{*}{$\leq 15$} 
 & 25  & \textbf{83.11}\%   & 72.35\%           \\ \cline{2-4}
 & 50  & \textbf{90.87}\%   & 83.42\%           \\ \cline{2-4}
 & 75  & \textbf{88.99}\%   & 85.96\%           \\ \cline{2-4}
 & 100 & \textbf{94.48}\%   & 89.88\%           \\ \hline
\multirow{4}{*}{$\leq 19$} 
 & 25  & \textbf{84.70}\%   & 83.21\%          \\ \cline{2-4}
 & 50  & \textbf{83.86}\%   & 65.00\%          \\ \cline{2-4}
 & 75  & \textbf{81.40}\%   & 80.60\%          \\ \cline{2-4}
 & 100 & \textbf{80.86}\%   & 79.07\%          \\ \hline
\end{tabularx}
}
\end{table}

\section{Related Work}\label{sec:related}

\subsection{Learning Timed Automata} Automata learning techniques seek to
construct a state-machine representation of timed behaviors. These approaches
are typically divided into \emph{active} and \emph{passive} learning.

\paragraph{Active learning} 

Active learning assumes the existence of an oracle or ``teacher'' that can
answer membership and equivalence queries during the learning process. 
Foundational work~\cite{grinchtein2010learning} adapts Angluin's classical $L^*$
algorithm to infer \emph{event-recording automata} (\textsf{ERA}s)~\cite{alur1999event}.
Subsequent research has focused on expanding the
expressiveness of the learned models while managing the inherent complexity of
continuous time. For instance,~\cite{henry2020active} extends active learning to timed
automata with unobservable clock resets. Some other works restrict the learning
problem to one-clock formalisms, such as one-clock timed automata
\cite{an2020learning, xu2022active} and single-timer Mealy machines
\cite{vaandrager2021learning}. 
More recently, significant advances have been made toward learning general multi-clock
\emph{deterministic timed automata} (\textsf{DTA}s)~\cite{waga2023active, teng2024learning}. Some other approaches have sought to optimise the
learning process by bridging passive state-merging with active query refinement
\cite{aichernig2020passive}, utilising model-based mutation testing to
iteratively refine timed automata from tests \cite{tappler2019time}, or
inferring symbolic timed constraints directly from concrete timed interactions
\cite{dierl2023learning}.

\paragraph{Passive learning} In contrast, passive learning operates solely on
given sets of positive (normal) and negative (anomalous) traces; the dominant methodology here is state-merging.
For example,~\cite{verwer2012efficiently} uses statistical heuristics on timestamped labels to efficiently identify \emph{deterministic real-time automata} (\textsf{DRTA}s), i.e.~\textsf{DTA}s with only one clock that resets on every transition
over \emph{integer-time} traces.
Another recent approach is the TAG algorithm~\cite{cornanguer2022tag},
which infers \textsf{DRTA}s from only positive traces.
Other approaches formulate passive timed automata
learning as an exact logical constraint satisfaction problem solved via SMT
\cite{tappler2022timed, majumdar2025learning}. We remark that passive learning of automata is generally computationally hard~(see e.g.,~\cite{verwer2011efficiency}).

\subsection{Learning Untimed $\ltl{}$ Specifications}

Texada~\cite{lemieux2015general} is the first tool to support passive learning of full $\ltl{}$. Subsequent state-of-the-art tools have employed sophisticated
combinatorial search techniques. \texttt{SCARLET}~\cite{raha2022scalable}, for instance, focusses on specifically
tailored fragments of $\ltl{}$ to improve scalability.
\cite{valizadeh2024gpu}~proposes a GPU-accelerated enumerative
algorithm that utilises characteristic tables and observational equivalence for
fast formula evaluation. Building upon this, \textsc{Bolt}~\cite{bathie2026ltl} is a CPU-based tool that achieves further orders-of-magnitude
speedups.

\subsection{Learning Timed Temporal Logics}

\paragraph{Parameter synthesis} Much of the existing literature relies on
template-based \emph{parameter synthesis}, where the structural architecture of the
formula is provided \textit{a priori}: the learning algorithm optimises
the numerical time bounds and signal thresholds. This is frequently formulated as a
statistical optimisation problem, where algorithms maximise quantitative
robustness semantics to classify stochastic processes~\cite{bartocci2014data},
detect cyber-physical anomalies~\cite{jones2014anomaly, kong2016temporal}, or
extract reward functions for reinforcement learning frameworks
\cite{li2017reinforcement, xu2019transfer, raha2023synthesizing}. While
template-based methods are less computationally demanding, they inherently limit the
expressive capability of the learned specifications to the designer's prior
domain assumptions.

\paragraph{Structure synthesis} Overcoming the template restriction requires
\emph{structure synthesis}---a significantly more complex problem due to the infinite
search space of formula architectures. 
Foundational work~\cite{godskesen1995synthesizing}
established the synthesis of distinguishing formulae for timed bisimulation
equivalence. Modern template-free approaches
tackle the computational challenge using diverse techniques. Some rely on systematic
enumerative search, coupled with aggressive heuristic pruning or partial
ordering to manage the combinatorial explosion \cite{kong2014temporal,
mohammadinejad2020interpretable}. Others adapt classical machine learning
paradigms, such as inferring interval temporal logic decision trees
\cite{brunello2019interval} or clustering execution traces to generate hybrid
system assertions \cite{nicoletti2024mining}. A notable recent approach~\cite{fronda2022differentiable}
maps discrete logic into continuous domains for differentiable neural-network
inference.

\paragraph{Comparison to our approach} While recent template-free structure
synthesis methods have shown considerable progress, they often face severe
scalability bottlenecks or expressiveness limits. For instance,~\cite{wang2025synthesis}
tackles the passive synthesis of (intersection- and renaming-free) \emph{timed regular expressions} ($\tre{}$s) by enumerating
untimed parametric $\tre{}$ templates and subsequently resolving the timing
constraints via SMT solvers. Although effective and able to find minimal (intersection- and renaming-free) $\tre{}$s, this
method is inherently tied to the rigid structure of that limited fragment. Similarly,~\cite{raha2023synthesizing} proposes an exact SMT-based approach to synthesise $\mtl{}$ formulae of arbitrary structure by encoding the entire syntax DAG
and continuous monitoring semantics into Linear Real Arithmetic. However, their reliance on a monolithic SMT
encoding to discover both the logical structure and timing bounds simultaneously
poses a substantial scalability challenge. 
Their focus is also quite different: finding minimal formulae with bounded future reach (lookahead) to prioritise online monitoring efficiency.

In contrast, our framework distinctively separates the timed and untimed aspects
of the problem. By extracting quantitative timing differences \textit{a priori}
and embedding them into the alphabet as atomic propositions, we bypass the need
for monolithic SMT structural encoding or template enumeration entirely. This
novel abstraction allows us to delegate the complex combinatorial search to
highly optimised untimed $\ltl{}$ tools (such as \textsc{Bolt} or \texttt{SCARLET}), guaranteeing
completeness while retaining the full expressiveness of $\mitppl{}$. 

\section{Conclusion and Future Work}\label{sec:conclusion}
We presented a framework that reduces timed specification mining to untimed $\ltl{}$ learning by
synthesising timed features to resolve propositional indistinguishability. The approach combines
collision detection, template injection, layered feature synthesis, and proposition minimisation,
with formal guarantees of soundness, progress, and completeness, alongside strong experimental
results across diverse real-time systems. Our experimental evaluation demonstrates its efficiency
and practical utility across diverse real-time scenarios.

Several directions remain for future work, including synthesising minimal timed formulae through
incremental enumeration, integrating regular expressions to better capture cyclic behaviours, and
reversing the workflow by first extracting untimed structural templates and subsequently
synthesising timing parameters. Furthermore, while our current approach is highly effective, it relies
on the availability of both positive and negative examples to extract separating timing pivots and
mathematically bound the learned specification. Recognising that observing anomalous failure traces
can be difficult in safety-critical domains, we plan to adapt our framework to learn exclusively
from positive traces via one-class classification (see,~e.g.,~\cite{roy2023learning}). Because our architecture strictly separates the
extraction of timing features from the synthesis of logical structure, these modules can be updated
independently. By shifting our feature extraction module to calculate tightness metrics such as
bounding envelopes from positive delays--we can swap our downstream structural learner for one that
optimises for language minimality over those features, acting as a regulariser to prevent overfitting.

\section*{Acknowledgment}
The authors would like to acknowledge the use of Claude 4.6 (Anthropic) to assist with copy-editing and grammatical refinement during the preparation of this manuscript.

\bibliography{biblio}

\clearpage
\onecolumn
\appendix

\begin{center}

\bgroup 

\renewcommand{\thetable}{}

\makeatletter
\def\@dbstargarg{}
\makeatother

\bgroup\footnotesize\begin{center}
    \textsc{Formulae synthesised with and without proposition minimisation.}
\end{center}\egroup
\vskip 0.5ex
\centering
\scalebox{0.65}{
\def\arraystretch{1}
\setlength\tabcolsep{1.5mm}
\begin{tabularx}{\textwidth}{|c|c|>{\tiny\ttfamily\raggedright\arraybackslash}X|>{\tiny\ttfamily\raggedright\arraybackslash}X|}
\hline
\normalfont\normalsize\textbf{$\ell$} & \normalfont\normalsize\textbf{$n$} & \multicolumn{1}{c|}{\normalfont\normalsize\textbf{With Proposition Minimisation}} & \multicolumn{1}{c|}{\normalfont\normalsize\textbf{Without Proposition Minimisation}} \\ \hline
\multirow{4}{*}{\normalsize $\leq 10$}
 & \normalsize 25 &
(G ((X[!] ((F[0, 42](!(q))) -> (F[0, 29](!(q) \&\& !(r))))) -> (F[0, 42](!(q))))) \& (((X (X ((F[0, 29](!(q) \&\& !(r))) <-> (X (F[0, 29](!(q) \&\& !(r))))))) \& (X ((! (F[0, 29](!(q) \&\& !(r)))) R ((F[0, 42](!(q))) U (F[0, 29](!(q) \&\& !(r))))))) \textbar{} ((G ((F[0, 29](!(q) \&\& !(r))) \textbar{} (X (F[0, 42](!(q)))))) \& (! ((X (X (X (X[!] (F[0, 29](!(q) \&\& !(r))))))) U (F[0, 29](!(q) \&\& !(r)))))))
 &
(G ((X[!] (X (X[!] (X[!] (r))))) -> (F[0, 29](!(q) \&\& !(r))))) \& ((X ((! (F[0, 29](!(q) \&\& !(r)))) R ((F[0, 42](!(q))) U (F[0, 29](!(q) \&\& !(r)))))) \textbar{} ((F[0, 42](!(q))) \& (G ((q) <-> ((q) R ((q) U (F[0, 29](!(q) \&\& !(r)))))))))
 \\ \cline{2-4}
 & \normalsize 50 &
(((G ((F[0, 14](!(r))) -> (X (X (X[!] (true)))))) \& ((F[0, 14](!(r))) \textbar{} (G (X ((F[0, 14](!(r))) \textbar{} (X (F[0, 15](!(q))))))))) \textbar{} (((F[0, 14](!(r))) R (! (F[0, 15](!(q))))) \& ((F (! ((X[!] (X[!] (X (F[0, 14](!(r)))))) -> (F[0, 14](!(r)))))) \& (G (((F[0, 14](!(r))) R (X[!] (X (F[0, 15](!(q)))))) -> (F[0, 14](!(r)))))))) \& ((! (F[0, 15](!(q)))) \& ((X (F[0, 15](!(q)))) \textbar{} ((! ((F[0, 15](!(q))) R (X (F[0, 14](!(r)))))) <-> (! ((F[0, 15](!(q))) R (X (X ((F[0, 15](!(q))) R (F[0, 14](!(r)))))))))))
 &
(G (F (! (r)))) \& ((G ((F[0, 14](!(r))) <-> (r))) \textbar{} ((! ((r) R ((F[0, 15](!(q))) U (! (F[0, 14](!(r))))))) \& ((G ((F[0, 15](!(q))) -> ((X (X (F[0, 15](!(q))))) -> (r)))) \textbar{} ((G ((X[!] (p)) -> (F[0, 14](!(r))))) \& ((! (F[0, 15](!(q)))) \& (G (((F[0, 14](!(r))) U (X[!] (X (false)))) -> (r))))))))
 \\ \cline{2-4}
 & \normalsize 75 &
((((((X (F[0, 10](!(p)))) -> (F[0, 25](!(q) \&\& !(r)))) \& ((F ((F[0, 25](!(q) \&\& !(r))) <-> ((F (F[0, 4](True))) R (! (F[0, 25](!(q) \&\& !(r))))))) \textbar{} (((F[0, 25](!(q) \&\& !(r))) <-> (X (F[0, 4](True)))) \& (G (((F[0, 4](True)) R (F[0, 10](!(p)))) -> (X (X (F[0, 10](!(p)))))))))) \textbar{} ((X ((F[0, 10](!(p))) -> (F[0, 4](True)))) \& ((G ((F[0, 4](True)) -> (X (X (X[!] (X (F[0, 10](!(p))))))))) \& ((X (F[0, 25](!(q) \&\& !(r)))) <-> (F (! ((X[!] (F[0, 25](!(q) \&\& !(r)))) -> (F (F[0, 4](True)))))))))) \& (((F[0, 10](!(p))) -> (F (X[!] ((F[0, 10](!(p))) \textbar{} (F[0, 4](True)))))) <-> ((X (F[0, 4](True))) \textbar{} ((F[0, 10](!(p))) R (! ((F[0, 4](True)) R (X[!] (F[0, 25](!(q) \&\& !(r)))))))))) \textbar{} (((F[0, 25](!(q) \&\& !(r))) <-> (X (F[0, 4](True)))) \& ((! ((X (F[0, 25](!(q) \&\& !(r)))) \textbar{} (F (F[0, 4](True))))) \textbar{} ((F (! ((F[0, 10](!(p))) -> (F (F[0, 4](True)))))) \& (X (X ((! (F[0, 25](!(q) \&\& !(r)))) R (F (F[0, 25](!(q) \&\& !(r))))))))))) \& ((X (X ((F[0, 10](!(p))) -> (F[0, 25](!(q) \&\& !(r)))))) \& ((F[0, 4](True)) \textbar{} (((X (F[0, 10](!(p)))) <-> ((X (X (F[0, 4](True)))) R (F[0, 10](!(p))))) \& ((F[0, 10](!(p))) \textbar{} (X (X ((F[0, 10](!(p))) R ((F[0, 4](True)) -> (F[0, 25](!(q) \&\& !(r)))))))))))
 &
((((X (! ((F (F[0, 4](True))) U (F[0, 10](!(p)))))) \textbar{} ((X ((F[0, 4](True)) -> (F[0, 25](!(q) \&\& !(r))))) <-> ((X (X ((F[0, 10](!(p))) -> (F[0, 25](!(q) \&\& !(r)))))) <-> ((X ((F[0, 10](!(p))) -> (F[0, 4](True)))) \& (((G (! ((F[0, 4](True)) R (p)))) \& (F ((F[0, 25](!(q) \&\& !(r))) <-> (X (p))))) \textbar{} ((((r) R (! (F[0, 4](True)))) \textbar{} (G ((F[0, 4](True)) -> ((F[0, 4](True)) U (p))))) \& ((G (! ((F[0, 4](True)) R (p)))) \textbar{} (! ((F (F[0, 4](True))) U ((F[0, 10](!(p))) \& (q))))))))))) \& ((F ((F[0, 10](!(p))) <-> ((F (F[0, 4](True))) <-> (F (p))))) <-> ((X (F[0, 4](True))) \textbar{} ((r) R ((F[0, 4](True)) -> (X (X[!] (r)))))))) \textbar{} (((F[0, 25](!(q) \&\& !(r))) <-> (X (F[0, 10](!(p))))) \& (((r) R (! (F[0, 10](!(p))))) \textbar{} ((G ((F[0, 25](!(q) \&\& !(r))) <-> (r))) \textbar{} (((r) R (F (F[0, 10](!(p))))) \& (G ((F[0, 4](True)) -> (X ((p) R (F (F[0, 25](!(q) \&\& !(r))))))))))))) \& ((G (F ((p) <-> (r)))) \& ((F[0, 25](!(q) \&\& !(r))) \textbar{} ((G (((X (F[0, 10](!(p)))) R (F[0, 4](True))) -> (X (p)))) <-> (X (X (((X (F[0, 4](True))) U (F[0, 10](!(p)))) -> (F[0, 25](!(q) \&\& !(r)))))))))
 \\ \cline{2-4}
 & \normalsize 100 &
(((X ((F[0, 18](!(r))) -> (X (q)))) \& ((G ((F (q)) <-> (X[!] (X[!] (true))))) \& (((((F[0, 18](!(r))) U (F[0, 7](True))) -> (F[0, 7](True))) \& (G ((F[0, 7](True)) -> ((q) -> (X (q)))))) \textbar{} ((X (X (X (X (F[0, 18](!(r))))))) \& (G ((F[0, 7](True)) -> ((q) -> (X (q))))))))) \textbar{} (G (X ((F[0, 7](True)) <-> (X[!] (q)))))) \& ((F ((F[0, 7](True)) \& ((q) <-> (X (F[0, 18](!(r))))))) \textbar{} (X ((q) U (X ((X[!] (true)) -> (F[0, 7](True)))))))
 &
(((((F[0, 40](!(q) \&\& !(r))) -> (F[0, 18](!(r)))) \& (G ((F[0, 7](True)) -> ((r) U (X (X[!] (F[0, 40](!(q) \&\& !(r))))))))) \& ((X ((F[0, 18](!(r))) -> (X (p)))) \& (F ((F[0, 7](True)) <-> ((p) R (X (F[0, 18](!(r))))))))) \textbar{} (((r) R (F (p))) \& (G ((X[!] ((F[0, 7](True)) R (X (r)))) -> (r))))) \& ((X (X (p))) \textbar{} ((F (F[0, 7](True))) R (! ((r) R (X[!] (F[0, 18](!(r))))))))
 \\ \hline
\multirow{4}{*}{\normalsize $\leq 15$}
 & \normalsize 25 &
(X (X ((F[0, 27](!(p) \&\& !(q) \&\& !(r))) U (X (X[!] (X[!] (true))))))) \& ((G (! ((F[0, 27](!(p) \&\& !(q) \&\& !(r))) \& (X (F[0, 27](!(p) \&\& !(q) \&\& !(r))))))) \textbar{} (((! (F[0, 27](!(p) \&\& !(q) \&\& !(r)))) R (F (F[0, 27](!(p) \&\& !(q) \&\& !(r))))) \& (((F[0, 27](!(p) \&\& !(q) \&\& !(r))) \& (X (X (X (F[0, 27](!(p) \&\& !(q) \&\& !(r))))))) \textbar{} (((F (F[0, 27](!(p) \&\& !(q) \&\& !(r)))) U (X (false))) \& ((X (! (F[0, 27](!(p) \&\& !(q) \&\& !(r))))) \textbar{} (G (! ((F[0, 27](!(p) \&\& !(q) \&\& !(r))) \& (X (X[!] (X[!] (F[0, 27](!(p) \&\& !(q) \&\& !(r))))))))))))))
 &
(F (! ((X[!] (F[0, 27](!(p) \&\& !(q) \&\& !(r)))) -> (p)))) \& ((G ((F[0, 27](!(p) \&\& !(q) \&\& !(r))) -> (X (X (F (F[0, 27](!(p) \&\& !(q) \&\& !(r)))))))) \textbar{} (((F[0, 27](!(p) \&\& !(q) \&\& !(r))) -> (X (X (p)))) \& ((G ((p) \textbar{} (X (X (F (F[0, 27](!(p) \&\& !(q) \&\& !(r)))))))) \textbar{} (G ((X[!] ((F[0, 27](!(p) \&\& !(q) \&\& !(r))) R (X (p)))) -> (r))))))
 \\ \cline{2-4}
 & \normalsize 50 &
((X (! (F[0, 8](!(p))))) \& ((G (! (F[0, 3](True)))) \textbar{} (((X (F[0, 3](True))) \textbar{} (X (X ((F[0, 3](True)) <-> (X (X[!] (X (false)))))))) \& ((X (! (F[0, 8](!(p))))) \& ((G (((F[0, 3](True)) R (F[0, 8](!(p)))) -> (X (F[0, 3](True))))) \textbar{} (G ((F[0, 3](True)) -> (X (X[!] (X[!] (X[!] (true)))))))))))) \& ((F ((F (F[0, 8](!(p)))) <-> (X (X (false))))) \& ((F[0, 3](True)) \textbar{} ((! ((F[0, 3](True)) R (X[!] ((F[0, 8](!(p))) U (X (F[0, 8](!(p)))))))) \& ((F[0, 8](!(p))) \textbar{} ((F[0, 3](True)) R ((F[0, 3](True)) -> ((F[0, 8](!(p))) -> (X (F[0, 8](!(p)))))))))))
 &
((G (F ((p) <-> (r)))) \& (((r) R (X (! (F[0, 8](!(p)))))) \textbar{} (((F[0, 3](True)) -> (F (X[!] (F[0, 3](True))))) <-> (G ((F[0, 3](True)) -> (X ((q) U (p)))))))) \& (((F[0, 3](True)) -> (F (X[!] (F[0, 8](!(p)))))) \& (G ((F[0, 8](!(p))) -> ((X (r)) -> (p)))))
 \\ \cline{2-4}
 & \normalsize 75 &
(((X (X (! (F[0, 3](True))))) \& ((G ((F[0, 3](True)) -> (X ((! (F[0, 3](True))) U (F[0, 34](!(q) \&\& !(r))))))) \& ((G (X (! (F[0, 6](True))))) \textbar{} ((F (! ((F[0, 6](True)) -> (F[0, 3](True))))) \& (G ((F[0, 6](True)) -> (X (X (F (F[0, 34](!(q) \&\& !(r)))))))))))) \textbar{} ((G ((F[0, 6](True)) -> (X (X[!] (X[!] (X (F[0, 34](!(q) \&\& !(r))))))))) \textbar{} ((F (! ((F[0, 3](True)) -> (F[0, 34](!(q) \&\& !(r)))))) \& (F (! ((X[!] (F[0, 34](!(q) \&\& !(r)))) -> (F (F[0, 6](True))))))))) \& (((F[0, 34](!(q) \&\& !(r))) -> (F (F[0, 6](True)))) <-> (((F[0, 34](!(q) \&\& !(r))) -> (X (X (X (F[0, 34](!(q) \&\& !(r))))))) \& ((G ((F[0, 3](True)) -> (X ((F[0, 34](!(q) \&\& !(r))) -> (X (F[0, 34](!(q) \&\& !(r)))))))) \& (((F[0, 3](True)) -> (G ((F[0, 6](True)) -> (F (F[0, 3](True)))))) <-> ((F[0, 3](True)) -> (((! (F[0, 6](True))) U (F[0, 3](True))) U (F[0, 34](!(q) \&\& !(r)))))))))
 &
(((X (X ((F[0, 3](True)) <-> (G (F (p)))))) \& ((F (! ((F[0, 34](!(q) \&\& !(r))) -> ((p) -> (F (F[0, 3](True))))))) \textbar{} (G ((r) \textbar{} (X ((X (F[0, 3](True))) -> (F[0, 6](True)))))))) \textbar{} ((F ((F[0, 6](True)) \& (p))) \& (G ((F[0, 6](True)) -> (X ((p) U (X[!] (r)))))))) \& ((G (! ((F[0, 3](True)) R ((F[0, 3](True)) U (r))))) \& ((F ((F[0, 3](True)) <-> ((q) <-> (F (p))))) \& ((F[0, 34](!(q) \&\& !(r))) -> ((F (F[0, 6](True))) <-> (F (X[!] (p)))))))
 \\ \cline{2-4}
 & \normalsize 100 &
((((((F[0, 10](!(p))) <-> (F[0, 40](!(q) \&\& !(r)))) R (X[!] (true))) \& ((X ((F[0, 10](!(p))) -> (X (X (X (F[0, 2](True))))))) \& (((F (F[0, 40](!(q) \&\& !(r)))) U (X (false))) R (! (F[0, 2](True)))))) \& ((G ((F[0, 10](!(p))) -> ((! (F[0, 2](True))) R (F[0, 10](!(p)))))) \& (X (X ((F[0, 2](True)) <-> (X (X[!] (X[!] (F[0, 2](True)))))))))) \textbar{} ((! ((F[0, 2](True)) R ((F[0, 2](True)) -> (F[0, 10](!(p)))))) R (F[0, 10](!(p))))) \& ((F (! ((X[!] (F[0, 40](!(q) \&\& !(r)))) -> (F (F[0, 2](True)))))) \& ((X (((F[0, 10](!(p))) <-> (X (F[0, 2](True)))) U (X (F[0, 10](!(p)))))) \& (((! (F[0, 40](!(q) \&\& !(r)))) \textbar{} (F ((F[0, 40](!(q) \&\& !(r))) \& (X (X[!] (X (F (F[0, 2](True))))))))) \& ((X (X (X (F[0, 10](!(p)))))) \textbar{} (F (! ((F[0, 10](!(p))) \textbar{} (X (X (X[!] (F[0, 40](!(q) \&\& !(r)))))))))))))
 &
(((G (! (F[0, 10](!(p))))) \textbar{} (((F[0, 40](!(q) \&\& !(r))) R (F (p))) \& ((! (F[0, 2](True))) \& ((X ((F[0, 10](!(p))) -> (X (p)))) \& (G ((X[!] ((F[0, 10](!(p))) \& (F[0, 2](True)))) -> (r))))))) \textbar{} ((F ((F[0, 10](!(p))) \& (r))) \& ((F[0, 2](True)) R ((r) R ((F[0, 2](True)) <-> (X[!] (p))))))) \& ((G ((F[0, 10](!(p))) -> ((X (r)) -> (p)))) \& (X (! ((F[0, 10](!(p))) U (G (X (F (r))))))))
 \\ \hline
\multirow{4}{*}{\normalsize $\leq 19$}
 & \normalsize 25 &
(X (! ((F[0, 25](!(p) \&\& !(q) \&\& !(r))) U (X (false))))) \& ((X (X ((F[0, 25](!(p) \&\& !(q) \&\& !(r))) \textbar{} (X (X[!] (true)))))) \& ((G ((F[0, 25](!(p) \&\& !(q) \&\& !(r))) -> ((F[0, 25](!(p) \&\& !(q) \&\& !(r))) U (X (F[0, 25](!(p) \&\& !(q) \&\& !(r))))))) \textbar{} (((F (F[0, 25](!(p) \&\& !(q) \&\& !(r)))) U (X (false))) \& (X (X ((F[0, 25](!(p) \&\& !(q) \&\& !(r))) \textbar{} (X (F[0, 25](!(p) \&\& !(q) \&\& !(r))))))))))
 &
(X (! ((F[0, 25](!(p) \&\& !(q) \&\& !(r))) U (X (false))))) \& ((G ((X[!] ((p) R (X (r)))) -> (F[0, 25](!(p) \&\& !(q) \&\& !(r))))) \textbar{} (G ((p) -> (X[!] (X[!] ((r) -> (F[0, 25](!(p) \&\& !(q) \&\& !(r)))))))))
 \\ \cline{2-4}
 & \normalsize 50 &
(((X (! ((F[0, 12](!(p))) U (X (false))))) \& (((F[0, 12](!(p))) U (X (X (X (X (false)))))) \textbar{} ((F (! ((F[0, 12](!(p))) \textbar{} (X (X (F (F[0, 12](!(p))))))))) \& (((X (X (F[0, 12](!(p))))) <-> ((F[0, 12](!(p))) \& (X ((X (X (X (F[0, 12](!(p)))))) R (! (F[0, 12](!(p)))))))) \textbar{} ((X (X (X (X ((F (F[0, 12](!(p)))) -> (F[0, 12](!(p)))))))) \& (X (! ((F[0, 12](!(p))) R (X[!] (X (X (F[0, 12](!(p)))))))))))))) \textbar{} ((((F (F[0, 12](!(p)))) U (X (false))) \& (X (! (((F[0, 12](!(p))) R (X (F[0, 12](!(p))))) -> (F[0, 12](!(p))))))) \textbar{} (((F[0, 12](!(p))) <-> (X (X (F[0, 12](!(p)))))) \& ((X (X (X (X[!] (X (X (X (F[0, 12](!(p)))))))))) \& ((X ((F[0, 12](!(p))) U (X[!] (X[!] (X[!] (F[0, 12](!(p)))))))) <-> (G (((F[0, 12](!(p))) R (X[!] (F[0, 12](!(p))))) -> (F[0, 12](!(p)))))))))) \& (((G ((F[0, 12](!(p))) -> (X (X[!] (true))))) \textbar{} ((F[0, 12](!(p))) \textbar{} (X (X ((X (F[0, 12](!(p)))) -> (F[0, 12](!(p)))))))) <-> ((G ((F[0, 12](!(p))) \textbar{} (X (X (X (X[!] (true))))))) \textbar{} (((X (F[0, 12](!(p)))) \textbar{} (X (! ((F[0, 12](!(p))) R (X (X[!] (X (F[0, 12](!(p)))))))))) <-> (((F[0, 12](!(p))) <-> (X ((F[0, 12](!(p))) \textbar{} (X (X (F[0, 12](!(p)))))))) \textbar{} ((F[0, 12](!(p))) <-> (G ((F[0, 12](!(p))) -> (X (X (X[!] (true)))))))))))
 &
(G (F (! (r)))) \& ((G (((F[0, 12](!(p))) R (q)) -> (X[!] (X (F[0, 12](!(p))))))) \textbar{} ((F (! ((F[0, 12](!(p))) -> (p)))) \& ((G ((q) -> ((F[0, 12](!(p))) <-> (p)))) \textbar{} (((X (X (F[0, 12](!(p))))) -> (F[0, 12](!(p)))) \& ((G (F ((p) <-> (r)))) \& (((r) R ((q) <-> (X (F[0, 12](!(p)))))) \textbar{} (((F[0, 12](!(p))) U ((! (F[0, 12](!(p)))) U (r))) \& (F (! ((F[0, 12](!(p))) \textbar{} (X (F (q)))))))))))))
 \\ \cline{2-4}
 & \normalsize 75 &
((((G ((F[0, 26](!(p) \&\& !(q) \&\& !(r))) \textbar{} (X (F (F[0, 6]((p))))))) \textbar{} ((G ((X[!] (X (X[!] (F[0, 6]((p)))))) -> (F[0, 26](!(p) \&\& !(q) \&\& !(r))))) \textbar{} (((F (F[0, 26](!(p) \&\& !(q) \&\& !(r)))) U (X (false))) \& (G ((F[0, 26](!(p) \&\& !(q) \&\& !(r))) -> (X ((X[!] (F[0, 26](!(p) \&\& !(q) \&\& !(r)))) -> (F[0, 26](!(p) \&\& !(q) \&\& !(r)))))))))) \& (((! (F[0, 26](!(p) \&\& !(q) \&\& !(r)))) U (X (X (F[0, 6]((p)))))) \textbar{} ((X (X (X ((X (F[0, 26](!(p) \&\& !(q) \&\& !(r)))) -> (F[0, 6]((p))))))) <-> ((F[0, 26](!(p) \&\& !(q) \&\& !(r))) \textbar{} ((F (! ((F[0, 26](!(p) \&\& !(q) \&\& !(r))) -> (X (X (F (F[0, 6]((p))))))))) <-> ((F (F[0, 6]((p)))) <-> ((! (F[0, 26](!(p) \&\& !(q) \&\& !(r)))) U ((F[0, 26](!(p) \&\& !(q) \&\& !(r))) U (X (X (F[0, 6]((p))))))))))))) \textbar{} ((X (X (! (F[0, 6]((p)))))) \& (((F[0, 6]((p))) R (F (F[0, 26](!(p) \&\& !(q) \&\& !(r))))) \& ((X ((X (F[0, 26](!(p) \&\& !(q) \&\& !(r)))) -> (F[0, 26](!(p) \&\& !(q) \&\& !(r))))) <-> (((F[0, 6]((p))) R (X (! ((F[0, 26](!(p) \&\& !(q) \&\& !(r))) <-> (F[0, 6]((p))))))) \textbar{} ((G ((F[0, 6]((p))) -> (X (X (F[0, 26](!(p) \&\& !(q) \&\& !(r))))))) \& (((F (F[0, 26](!(p) \&\& !(q) \&\& !(r)))) U (X (false))) U (F[0, 6]((p)))))))))) \& ((G ((F[0, 26](!(p) \&\& !(q) \&\& !(r))) -> (X (F[0, 6]((p)))))) \textbar{} (((F (F[0, 26](!(p) \&\& !(q) \&\& !(r)))) U (X (false))) \& ((! ((F[0, 6]((p))) U (F[0, 26](!(p) \&\& !(q) \&\& !(r))))) \textbar{} (((F[0, 6]((p))) <-> ((F[0, 26](!(p) \&\& !(q) \&\& !(r))) U (X (false)))) \& (F ((F[0, 6]((p))) <-> ((F[0, 26](!(p) \&\& !(q) \&\& !(r))) \textbar{} (X (X (F[0, 26](!(p) \&\& !(q) \&\& !(r))))))))))))
 &
((((G ((p) -> ((! (F[0, 26](!(p) \&\& !(q) \&\& !(r)))) R (F (F[0, 26](!(p) \&\& !(q) \&\& !(r))))))) \textbar{} ((G ((p) -> (F (F[0, 26](!(p) \&\& !(q) \&\& !(r)))))) \& (! ((F[0, 26](!(p) \&\& !(q) \&\& !(r))) R ((F[0, 26](!(p) \&\& !(q) \&\& !(r))) -> (p)))))) \& ((G (X ((p) -> (F[0, 26](!(p) \&\& !(q) \&\& !(r)))))) \textbar{} (((G ((X[!] ((F[0, 26](!(p) \&\& !(q) \&\& !(r))) R (X (q)))) -> (r))) \textbar{} (G (X ((p) -> ((! (F[0, 6]((p)))) U (F[0, 26](!(p) \&\& !(q) \&\& !(r)))))))) \& ((F[0, 6]((p))) \textbar{} ((G ((q) U (((F[0, 26](!(p) \&\& !(q) \&\& !(r))) R (p)) -> (F[0, 6]((p)))))) \textbar{} (G ((p) -> ((! (F[0, 26](!(p) \&\& !(q) \&\& !(r)))) R (F (F[0, 26](!(p) \&\& !(q) \&\& !(r)))))))))))) \textbar{} (((F[0, 6]((p))) R (F (F[0, 26](!(p) \&\& !(q) \&\& !(r))))) \& ((G ((r) \textbar{} (X ((X (p)) -> (F[0, 6]((p))))))) \textbar{} (((r) R (! (F[0, 6]((p))))) \& (F (! (((F[0, 6]((p))) U (F[0, 26](!(p) \&\& !(q) \&\& !(r)))) -> (F[0, 26](!(p) \&\& !(q) \&\& !(r)))))))))) \& ((F (! ((p) -> (F[0, 26](!(p) \&\& !(q) \&\& !(r)))))) <-> (! (((F[0, 6]((p))) \& (p)) R (X (F[0, 26](!(p) \&\& !(q) \&\& !(r)))))))
 \\ \cline{2-4}
 & \normalsize 100 &
(((((X (F[0, 46](!(q) \&\& !(r)))) <-> (((F[0, 20](!(p) \&\& !(r))) U (X (X (F[0, 20](!(p) \&\& !(r)))))) <-> ((X ((F[0, 20](!(p) \&\& !(r))) -> (X (X (F[0, 20](!(p) \&\& !(r))))))) \& ((G ((F[0, 20](!(p) \&\& !(r))) \textbar{} (X (F[0, 20](!(p) \&\& !(r)))))) \textbar{} (G ((F[0, 20](!(p) \&\& !(r))) -> (X (X (F (F[0, 46](!(q) \&\& !(r)))))))))))) \textbar{} ((G (X ((F[0, 20](!(p) \&\& !(r))) -> (F[0, 46](!(q) \&\& !(r)))))) \& ((G ((F[0, 20](!(p) \&\& !(r))) -> (X (X (F (F[0, 20](!(p) \&\& !(r)))))))) \textbar{} ((X (! (F[0, 20](!(p) \&\& !(r))))) \& (X (X (X (X (X[!] (X (! (F[0, 20](!(p) \&\& !(r)))))))))))))) \& ((G ((F[0, 20](!(p) \&\& !(r))) -> (F (F[0, 46](!(q) \&\& !(r)))))) \& ((! ((F[0, 20](!(p) \&\& !(r))) -> (X (F[0, 20](!(p) \&\& !(r)))))) \textbar{} (G ((F[0, 20](!(p) \&\& !(r))) -> (X (X (X[!] (true))))))))) \textbar{} ((X (! ((F[0, 46](!(q) \&\& !(r))) U ((F[0, 20](!(p) \&\& !(r))) <-> (X (F[0, 20](!(p) \&\& !(r)))))))) \textbar{} (((X (! (F[0, 20](!(p) \&\& !(r))))) \& ((G ((F[0, 20](!(p) \&\& !(r))) <-> (F[0, 46](!(q) \&\& !(r))))) R (X (X (F[0, 46](!(q) \&\& !(r))))))) \textbar{} ((G ((F[0, 20](!(p) \&\& !(r))) -> (X (X (X (X[!] (true))))))) \& (G ((F[0, 46](!(q) \&\& !(r))) -> (X ((F[0, 20](!(p) \&\& !(r))) U (X (F[0, 20](!(p) \&\& !(r)))))))))))) \& (((((F[0, 20](!(p) \&\& !(r))) U (((F[0, 20](!(p) \&\& !(r))) R (F[0, 46](!(q) \&\& !(r)))) R (! (F[0, 20](!(p) \&\& !(r)))))) \& ((G ((F[0, 20](!(p) \&\& !(r))) \textbar{} (X (F[0, 20](!(p) \&\& !(r)))))) \textbar{} ((F (! ((F[0, 20](!(p) \&\& !(r))) -> (F[0, 46](!(q) \&\& !(r)))))) \textbar{} ((F[0, 20](!(p) \&\& !(r))) \& (X ((F[0, 20](!(p) \&\& !(r))) <-> (X ((F[0, 20](!(p) \&\& !(r))) <-> (X (F[0, 20](!(p) \&\& !(r)))))))))))) \textbar{} ((((F[0, 46](!(q) \&\& !(r))) U (X (false))) \& (X ((F[0, 20](!(p) \&\& !(r))) R (X (X (X (X[!] (F[0, 20](!(p) \&\& !(r)))))))))) \textbar{} ((X (X (X (X[!] (X (X (X (F[0, 20](!(p) \&\& !(r)))))))))) \& ((F[0, 20](!(p) \&\& !(r))) U (G ((F[0, 20](!(p) \&\& !(r))) -> (X (X (F[0, 46](!(q) \&\& !(r))))))))))) \textbar{} (X (! ((G ((F[0, 20](!(p) \&\& !(r))) -> (F[0, 46](!(q) \&\& !(r))))) U (F[0, 20](!(p) \&\& !(r)))))))
 &
(((F (F[0, 46](!(q) \&\& !(r)))) \& ((G (X ((p) -> ((F[0, 20](!(p) \&\& !(r))) U (X[!] (F[0, 46](!(q) \&\& !(r)))))))) \& ((G ((! (p)) U ((q) <-> (X[!] (F[0, 20](!(p) \&\& !(r))))))) \textbar{} (((! (F[0, 20](!(p) \&\& !(r)))) R (F (F[0, 20](!(p) \&\& !(r))))) \& ((((q) U (F[0, 20](!(p) \&\& !(r)))) U (X (false))) \textbar{} (G (((F[0, 20](!(p) \&\& !(r))) U (X[!] (X (false)))) -> (r)))))))) \textbar{} ((F[0, 20](!(p) \&\& !(r))) \& (((r) R ((q) <-> (X (F[0, 20](!(p) \&\& !(r)))))) \textbar{} (((r) R ((p) <-> (X (X (X (F[0, 20](!(p) \&\& !(r)))))))) \textbar{} (((X (! (F[0, 20](!(p) \&\& !(r))))) \& (G ((p) \textbar{} (X (X (F[0, 20](!(p) \&\& !(r)))))))) \textbar{} ((G ((p) -> (F (F[0, 46](!(q) \&\& !(r)))))) \& (! ((F[0, 46](!(q) \&\& !(r))) U (X (X ((p) R (F[0, 46](!(q) \&\& !(r)))))))))))))) \& ((((X (F[0, 20](!(p) \&\& !(r)))) \textbar{} (X ((F[0, 46](!(q) \&\& !(r))) U (! ((q) <-> (X (F[0, 20](!(p) \&\& !(r))))))))) \& (((G ((p) -> (F[0, 20](!(p) \&\& !(r))))) \textbar{} (X (((F[0, 20](!(p) \&\& !(r))) R (q)) -> (F[0, 20](!(p) \&\& !(r)))))) <-> ((F (! (((F[0, 20](!(p) \&\& !(r))) R (F[0, 46](!(q) \&\& !(r)))) -> (F[0, 20](!(p) \&\& !(r)))))) \& ((G ((X[!] (r)) -> (p))) \textbar{} (((F[0, 46](!(q) \&\& !(r))) R (F (p))) \& ((X (((F[0, 20](!(p) \&\& !(r))) R (q)) -> (F[0, 20](!(p) \&\& !(r))))) \& (G ((r) -> (X[!] ((F[0, 46](!(q) \&\& !(r))) <-> (F (F[0, 20](!(p) \&\& !(r)))))))))))))) \textbar{} ((X (! ((G ((F[0, 20](!(p) \&\& !(r))) -> (F[0, 46](!(q) \&\& !(r))))) U (F[0, 20](!(p) \&\& !(r)))))) \textbar{} ((F[0, 20](!(p) \&\& !(r))) \& (X (! (((p) R (F[0, 46](!(q) \&\& !(r)))) R (q)))))))
 \\ \hline
\end{tabularx}
}
\egroup
\end{center}

\end{document}